%% file: main.tex
\documentclass[authoryear]{elsarticle}
\renewcommand{\cite}[1]{\citep{#1}}

\usepackage{amssymb}
\usepackage{amsmath}
\usepackage{amsthm}
\usepackage{stmaryrd}
\usepackage[usenames,dvipsnames]{xcolor}
\usepackage[hidelinks]{hyperref}
\usepackage{proof-dashed}
\usepackage{tikz}
\usetikzlibrary{arrows}
\usetikzlibrary{positioning}
\usepackage{import}
\usepackage{wrapfig}
\usepackage{breakcites}
\usepackage[all]{xy}

\newtheorem{theorem}{Theorem}
\newtheorem{definition}{Definition}
\newtheorem{lemma}{Lemma}
\newtheorem{example}{Example}
\newtheorem{corollary}{Corollary}

\newcommand{\newtext}{\color{black}}
\newcommand{\oldtext}{\color{black}}

\input{macros}

\begin{document}
\pagestyle{plain}



\title{A Core Model for Choreographic Programming}
\author{Lu\'\i s Cruz-Filipe}
\ead{lcf@imada.sdu.dk}
\author{Fabrizio Montesi}
\ead{fmontesi@imada.sdu.dk}

\address{University of Southern Denmark, Department of Mathematics and Computer Science, Campusvej 55, 5230 Odense M, 
Denmark}

\begin{abstract}
\input{abstract}
\end{abstract}


\maketitle

\section{Introduction}
\label{sec:intro}
\input{intro}

\section{Core Choreographies and Minimal Choreographies}
\label{sec:cc}
\input{cc}

\section{Stateful Processes, Minimal Processes and EndPoint Projections}
\label{sec:sp}
\input{sp}

\section{Turing Completeness of MC and Its Consequences}
\label{sec:rec}
\input{kleene}

\section{Removing Selections}
\label{sec:selection}
\input{selection}

\section{Minimality in Choreographies}
\label{sec:implications}
\input{minimality}

\section{Related Work and Discussion}
\label{sec:related}
\input{related}
%

\section*{Acknowledgements}
We thank Hugo Torres Vieira and Gianluigi Zavattaro for their useful comments.
Montesi was supported by CRC (Choreographies for Reliable and efficient Communication software), grant no.\ 
DFF--4005-00304 from the Danish Council for Independent Research.


\section*{References}

\bibliographystyle{elsarticle-harv}
\bibliography{biblio}

\appendix
\input{appendix}

\end{document}

%% file: macros.tex
\newcommand{\m}[1]{\ensuremath{\mathsf{#1}}}
\newcommand{\pid}[1]{\m{#1}}
\newcommand{\pids}[1]{\til {\pid #1}}
\newcommand{\role}[1]{\mathtt{#1}}
\newcommand{\roles}[1]{\tilde {\mathtt {#1}}}
\newcommand{\carr}[1]{\langle #1  \rangle}
\newcommand{\lgbox}[1]{\fcolorbox{lightgray}{white}{#1}}
\newcommand{\mlgbox}[1]{\fcolorbox{lightgray}{white}{\ensuremath{#1}}}

\newcommand{\emptyN}{\nil} 
\newcommand{\gseq}{\fatsemi}

\newcommand{\swapC}{\simeq_{C}}
\newcommand{\code}[1]{\texttt{#1}}
\newcommand{\nil}{\boldsymbol 0}
\newcommand{\res}[1]{\boldsymbol( \boldsymbol\nu #1 \boldsymbol)\:}
\newcommand{\com}[2]{#1\;\code{-\hspace{-0.3mm}>}\;#2}
\newcommand{\gencom}{\com{\pid p.e}{\pid q}}
\newcommand{\sel}[3]{\com{#1}{#2 [#3]}}
\newcommand{\gensel}{\sel{\pid p}{\pid q}{l}}
\newcommand{\lto}[1]{\mathrel{\stackrel{{\;\;#1\;\;}}{\mbox{\rightarrowfill}}}}
\newcommand{\lmto}[1]{\mathrel{\stackrel{#1}{\mbox{\rightarrowfill}}\!\!\raisebox{1ex}{\scriptsize$\ast$}}}
\newcommand{\lleft}{\textsc{l}}
\newcommand{\lright}{\textsc{r}}
\renewcommand{\merge}{\sqcup}
\newcommand{\cond}[3]{\m{if}\, #1 \, \m{then} \, #2 \, \m{else} \, #3}
\newcommand{\eqcom}[2]{#1 \! \stackrel{\code{<\!-}}{\code{=}}\! #2}

\newcommand{\gencond}{\cond{\eqcom{\pid p}{\pid q}}{C_1}{C_2}}

\newcommand{\tpid}[2]{\pid #1[\role #2]}
\newcommand{\scgencond}{\cond{\pid p.(e = e')}{C_1}{C_2}}
\newcommand{\scgencondv}{\cond{\pid p.(v = w)}{C_1}{C_2}}
\newcommand{\sccom}[3]{\com{#1}{#2} : #3}
\newcommand{\scsel}[4]{\sccom{#1}{#2}{#3[#4]}}
\newcommand{\scdel}[4]{\sccom{#1}{#2}{#3\carr{#4}}}
\newcommand{\scgensel}{\scsel{\tpid pA}{\tpid qB}{k}{l}}
\newcommand{\scgenselj}{\scsel{\tpid pA}{\tpid qB}{k}{l_j}}
\newcommand{\scgendel}{\scdel{\tpid pA}{\tpid qB}{k}{k'[\role C]}}
\newcommand{\scgencom}{\sccom{\tpid pA.e}{\tpid qB.x}{k}}
\newcommand{\scgencomv}{\sccom{\tpid pA.v}{\tpid qB.x}{k}}
\newcommand{\scstart}[4]{#1 \, \m{start} \, #2 : #3 (#4)}
\newcommand{\scgenstart}{\scstart{\wtil{\tpid pA}}{\wtil{\tpid qB}}{a}{k}}

\newcommand{\rec}[3]{\m{def} \, #1  =  #2 \, \m{in} \, #3}
\newcommand{\genrec}{\rec{X}{C_2}{C_1}}
\newcommand{\scgenrec}{\rec{X(\til D)}{C_2}{C_1}}
\newcommand{\pn}{\m{pn}}
\newcommand{\pnast}[1]{|{#1}|}
\newcommand{\call}[1]{#1}
\newcommand{\gencall}{X}
\definecolor{light-gray}{gray}{0.928}
\newcommand{\hl}[1]{\colorbox{light-gray}{\ensuremath{#1}}}
\newcommand{\tjudge}[3]{#1 \vdash \hl{#2} \triangleright #3}
\newcommand{\sccall}[2]{#1\carr{#2}}
\newcommand{\scgencall}{\sccall X{\til E}}
\newcommand{\pcont}{\mathtt{c}}
\newcommand{\til}{\tilde}
\newcommand{\wtil}{\widetilde}
\newcommand{\amend}{\m{Amend}}

\newcommand{\NN}{\mathbb N}
\newcommand{\enc}[3]{\{\!|#1|\!\}^{#2}_{#3}}
\newcommand{\altenc}[3]{\{\!\!\{#1\}\!\!\}^{#2}_{#3}}
\newcommand{\s}{\mathtt s}
\newcommand{\suc}{\mathtt s \cdot}
\newcommand{\Sum}{\mathsf{add}}
\newcommand{\Sub}{\mathsf{sub}}
\newcommand{\Eq}{\mathsf{eq}}

\newcommand{\epp}[2]{[\![#1]\!]_{#2}}
\newcommand{\RR}{\mathcal{R}}

\newcommand{\macro}[1]{\textsc{#1}}
\newcommand{\rname}[2]{\ensuremath{\mbox{{#1}$|${#2}}}}
\newcommand{\co}[1]{\overline{#1}}
\newcommand{\asend}[2]{{#1}!\carr{#2}}
\newcommand{\arecv}[1]{#1?}
\newcommand{\actor}[3]{#1 \triangleright_{#2} #3}
\newcommand{\parp}{\, \boldsymbol{|} \, }
\newcommand{\asel}[2]{{#1}\oplus#2}
\newcommand{\abranch}[2]{{#1}\&{#2}}
\newcommand{\precongr}{\preceq}
\newcommand{\defeq}{\stackrel{\Delta}{=}}
\newcommand{\zero}{\varepsilon}

\newcommand{\sstate}[1]{\left\{\!\small\begin{array}l #1\end{array}\!\right\}}

\newcommand{\CPP}{\ensuremath{\mathrm{SP}^\mathrm{CC}}}
\newcommand{\MPP}{\ensuremath{\mathrm{SP}^\mathrm{MC}}}
\newcommand{\ChPP}{\ensuremath{\mathrm{ChP}^\mathrm{ChC}}}

\newcommand{\gtcom}[3]{\com{#1}{#2}\!:\!\carr{#3}}
\newcommand{\gtgencom}{\gtcom{\role A}{\role B}{\tnat}}
\newcommand{\gtgenbranch}{\com{\role A}{\role B}:\{l_i:G_i\}_{i\in I}}
\newcommand{\gtrecvar}[1]{\mathbf{#1}}
\newcommand{\gtgenrec}{\mu \gtrecvar t}
\newcommand{\gtgencall}{\gtrecvar t}
\newcommand{\gtend}{\m{end}}
\newcommand{\tnat}{\mathbf{nat}}
\newcommand{\tstring}{\mathbf{string}}

\newcommand{\smallpar}[1]{\smallskip

\noindent {\emph{#1}}}

\newcommand{\cpreq}[3]{\co{#1}[#2](#3)}
\newcommand{\cpgenreq}{\cpreq{a}{\roles A}{k}}
\newcommand{\cpacc}[3]{#1[#2](#3)}
\newcommand{\cpgenacc}{\cpacc{a}{\role A}{k}}
\newcommand{\cpserv}[3]{!#1[#2](#3)}
\newcommand{\cpgenserv}{\cpserv{a}{\role A}{k}}
\newcommand{\cpsend}[4]{#1[#2]!#3\carr{#4}}
\newcommand{\cpgensend}{\cpsend {k}{\role A}{\role B}{e}}
\newcommand{\cpgensendv}{\cpsend {k}{\role A}{\role B}{v}}
\newcommand{\cprecv}[4]{#1[#2]?#3(#4)}
\newcommand{\cpgenrecv}{\cprecv {k}{\role B}{\role A}{x}}
\newcommand{\cpsel}[4]{#1[#2]!#3 \oplus #4}
\newcommand{\cpgensel}{\cpsel{k}{\role A}{\role B}{l}}
\newcommand{\cpgenselj}{\cpsel{k}{\role A}{\role B}{l_j}}
\newcommand{\cpbranch}[4]{#1[#2]?#3 \& #4}
\newcommand{\cpgenbranch}{\cpbranch{k}{\role B}{\role A}{\{ l_i : P_i \}_{i \in I}}}
\newcommand{\cpgenbranchQ}{\cpbranch{k}{\role B}{\role A}{\{ l_i : Q_i \}_{i \in I}}}

\newcommand{\cpgendels}{\cpsend {k}{\role A}{\role B}{k'[\role C]}}
\newcommand{\cpdelr}[4]{#1[#2]?#3(#4)}
\newcommand{\cpgendelr}{\cpdelr {k}{\role B}{\role A}{k'[\role C]}}
\newcommand{\cpgenrec}{\rec{X(\til x, \til k)}{Q}{P}}
\newcommand{\cpgencall}{\sccall{X}{\til e, \til k}}
\newcommand{\cpgencond}{\cond{e = e'}{P}{Q}}
\newcommand{\cpspar}{\ \; | \ \; }
\newcommand{\acspar}{\hspace{0.8mm} | \hspace{0.8mm} }
\newcommand{\seq}{\vdash}
\newcommand{\emb}[1]{\left\{\!\!\left[{#1}\right]\!\!\right\}}
\newcommand{\numeral}[1]{\ensuremath{\ulcorner{#1}\urcorner}}
\newcommand{\gencondmin}{\cond{\eqcom{\pid p}{\pid q}}{S(\pid p,\pids r,\lleft,C_1)}{S(\pid p,\pids r,\lright,C_2)}}
\newcommand{\delsel}[1]{(\!|{#1}|\!)}
\newcommand{\delselp}[2]{(\!|{#1},{#2}|\!)}
\newcommand{\pidb}[1]{\m{#1}^\bullet}

%% file: abstract.tex
Choreographic Programming is a paradigm for developing concurrent programs that are deadlock-free
by construction,
\newtext
as a result of
\oldtext
programming communications declaratively and then synthesising 
process implementations automatically. Despite strong interest on 
choreographies, a foundational model that explains which computations can be performed with the hallmark
constructs of choreographies is still missing.

In this work, we introduce Core Choreographies (CC), a
model that includes only the core primitives of choreographic programming.
Every computable function can be implemented as a choreography in CC, from which we can synthesise a
process implementation where independent computations run in parallel.
We discuss the design of CC and argue that it constitutes a canonical model for choreographic programming.
%

%% file: intro.tex
Programming concurrent and distributed systems is hard, because it is challenging to predict how 
programs executed at 
the same time in different computers will interact. Empirical studies reveal two important lessons: (i) while 
programmers have clear intentions about the order in which communication actions should be performed, tools do not 
adequately support them in translating these wishes to code~\cite{LPSZ08}; (ii) combining different communication 
protocols in a single application is a major source of mistakes~\cite{LLLG16}.

The paradigm of Choreographic Programming~\cite{M13:phd} was introduced to address these problems.
In this paradigm, programmers declaratively write the communications that they wish to take place, as programs called 
\emph{choreographies}.
Choreographies are descriptions of concurrent systems that syntactically disallow writing
mismatched I/O actions, inspired by the ``Alice and Bob'' notation of security protocols~\cite{NS78}.
An EndPoint Projection (EPP) can then be used to synthesise implementations in process models, which faithfully realise 
the communications given in the choreography and are guaranteed to be deadlock-free by construction 
even in the presence of arbitrary protocol compositions~\cite{CHY12,CM13}.


So far, work on choreographic programming focused on features of practical value -- including web
services~\cite{CHY12}, multiparty sessions~\cite{CM13,chor:website}, modularity~\cite{MY13}, and runtime
adaptation~\cite{DGGLM17}.
The models proposed all come with differing domain-specific syntaxes, semantics and EPP definitions (e.g., for
channel mobility or runtime adaptation), and cannot be considered minimal.
Another problem, arguably a consequence of the former, is that choreographic programming is meant for
implementation, but we still know little of what can be computed with the code obtained
from choreographies (\emph{choreography projections}).
The expressivity of the aforementioned models is evaluated just by showing some examples.

In this paper, we propose a canonical model for choreographic programming, called Core Choreographies (CC).
CC includes only the core primitives that can be found in most choreography languages, restricted to
the minimal requirements to achieve the computational power of Turing machines.
In particular, local computation at processes is severely restricted, and therefore nontrivial computations must be
implemented by using communications.
Therefore, CC is both representative of the paradigm and simple enough to analyse from a theoretical perspective.
Our technical development is based on a natural notion of function implementation, and the proof of Turing completeness
yields an algorithm for constructing a choreography that implements any given computable function.
Since choreographies describe concurrent systems, it is also natural to ask how much parallelism choreographies
exhibit. CC helps us in formally defining parallelism in choreographies; we exemplify how to use this notion to 
reason about the concurrent implementation of functions.
%

\begin{figure}
  \centering
  \resizebox{\textwidth}{!}{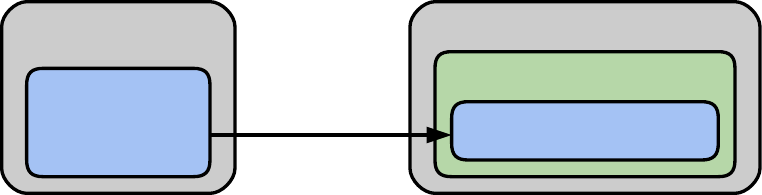}
  \label{fig:classes}
  \caption{Choreographic Programming}
\end{figure}

Yet, analysing the expressivity of choreographies is not enough.
What we are ultimately interested in is what can be computed with choreography projections, since those are
the terms that represent executable code.
However, the expressivity of choreographies does not translate directly to expressivity of projections, because EPP is
typically an incomplete procedure: it must guarantee deadlock-freedom, which in
previous models is obtained by complex requirements, e.g., type systems~\cite{CHY12,CM13}.
Therefore, only a subset of choreographies (projectable choreographies) can be used to synthesise process
implementations. The EPPs of such projectable choreographies form the set of choreography projections, which are
deadlock-free processes (see Figure~\ref{fig:classes}).

The main technical contribution of this paper is showing that the set of projectable choreographies in CC is still 
Turing complete. Therefore, by EPP, the set of corresponding choreography projections is also Turing complete, leading 
us to a characterisation of a Turing complete and deadlock-free fragment of a process calculus (which follows the same 
minimal design of CC). Furthermore, the parallel behaviour observed in CC choreographies for function implementations
translates directly to parallel execution of the projected processes.

More importantly, the practical consequence of our results is that CC is a
simple common setting for the study of foundational questions in choreographies.
This makes CC an appropriate foundational model for choreographic programming, akin to
$\lambda$-calculus for functional programming and $\pi$-calculus for mobile processes.
As an example of such foundational questions, we describe how the standard communication primitive
of label selection can be removed from CC without altering its computational power, yielding a truly
minimal choreography language wrt computation called Minimal Choreographies (MC).
However, doing so eliminates the clean separation between data and behaviour in message exchanges, which makes the 
resulting choreography hard to read.
\newtext
This result suggests that implementations of choreography languages -- 
like \cite{M13:phd}, \cite{DGGLM17}, and \cite{wscdl} -- may
adopt a simpler intermediate representation of choreographies that does not include label 
selections, by applying our elimination technique to the input choreography program 
(where programmers should be able to use selections, for readability).
A key advantage of having a simpler model without selections is that it bypasses the need for the 
standard notion of merging~\cite{CHY12}, which is typically one of the most complicated steps in 
EPP. We formally illustrate this point by showing that our EPP for MC enjoys an elegant definition.
\oldtext

\paragraph{Structure of the paper}
\S~\ref{sec:cc} defines Core Choreographies (CC) and its subcalculus of Minimal Choreographies (MC).
\S~\ref{sec:sp} introduces Stateful Processes (SP), our target process model, and its sublanguage of
Minimal Processes (MP), together with an EndPoint Projection (EPP) from CC to SP (and 
from MC to MP). We
also show that every unprojectable choreography in CC can be amended (transformed into a projectable one) by adding 
only label selections.
We prove that CC and its set of choreography projections are Turing complete in \S~\ref{sec:rec}.
In \S~\ref{sec:selection} we discuss label selections, and show that they can be encoded by
communications; this yields an amendment strategy for MC.
In \S~\ref{sec:implications}, we show that all the remaining primitives of CC are necessary to achieve
Turing completeness, and discuss the implications of our work for other choreography languages and process calculi -- 
in particular, we identify a Turing complete and deadlock-free fragment of value-passing CCS.
Related work and discussion are given in \S~\ref{sec:related}.

\newtext
\paragraph{Publication history}
This is an extended version of the material previously presented in~\cite{ourstuff}.
Noteworthy improvements include: detailed proofs of all results (in particular, 
Theorems~\ref{teo:sound1}, \ref{thm:delsel-wd}, \ref{lem:nocond-dec}, and \ref{lem:ifc-dec});
full definitions of our models (namely, the definitions of structural precongruence for CC and SP);
full definitions of EPP and the merge operator. Moreover, the encoding of selections into 
communications (\S~\ref{sec:selection}) and the encoding of CC into Channel Choreographies 
(\S~\ref{sec:translation}) are new results and have not been published previously.
\oldtext



%% file: classes_horizontal.pdf_tex
\begingroup%
  \makeatletter%
  \providecommand\color[2][]{%
    \errmessage{(Inkscape) Color is used for the text in Inkscape, but the package 'color.sty' is not loaded}%
    \renewcommand\color[2][]{}%
  }%
  \providecommand\transparent[1]{%
    \errmessage{(Inkscape) Transparency is used (non-zero) for the text in Inkscape, but the package 'transparent.sty' is not loaded}%
    \renewcommand\transparent[1]{}%
  }%
  \providecommand\rotatebox[2]{#2}%
  \ifx\svgwidth\undefined%
    \setlength{\unitlength}{365.599992bp}%
    \ifx\svgscale\undefined%
      \relax%
    \else%
      \setlength{\unitlength}{\unitlength * \real{\svgscale}}%
    \fi%
  \else%
    \setlength{\unitlength}{\svgwidth}%
  \fi%
  \global\let\svgwidth\undefined%
  \global\let\svgscale\undefined%
  \makeatother%
  \begin{picture}(1,0.2560175)%
    \put(0,0){\includegraphics[width=\unitlength]{classes_horizontal.pdf}}%
    \put(0.05685868,0.19910765){\color[rgb]{0,0,0}\makebox(0,0)[lb]{\smash{Choreographies}}}%
    \put(0.05685868,0.12252123){\color[rgb]{0,0,0}\makebox(0,0)[lb]{\smash{Projectable}}}%
    \put(0.61484562,0.21004855){\color[rgb]{0,0,0}\makebox(0,0)[lb]{\smash{Processes}}}%
    \put(0.61484562,0.14440303){\color[rgb]{0,0,0}\makebox(0,0)[lb]{\smash{Deadlock-free Processes}}}%
    \put(0.61484562,0.07875755){\color[rgb]{0,0,0}\makebox(0,0)[lb]{\smash{Choreography Projections}}}%
    \put(0.43979091,0.11158031){\color[rgb]{0,0,0}\makebox(0,0)[b]{\smash{EPP}}}%
    \put(0.05685868,0.07875755){\color[rgb]{0,0,0}\makebox(0,0)[lb]{\smash{Choreographies}}}%
  \end{picture}%
\endgroup%

%% file: cc.tex
We introduce Core Choreographies (CC), define function implementation and parallel execution of
choreographies, and prove some key properties of CC.

\subsection{Syntax of CC}
The syntax of CC is displayed in Figure~\ref{fig:cc_syntax},
where $C$ ranges over choreographies.
\input{cc_syntax}
We use two (infinite) disjoint sets of names:
processes ($\pid p,\pid q,\ldots$) and procedures
($X,\ldots$).
Processes run in parallel, and each process stores a
value -- a string of the form $\s\cdots\s\cdot\zero$ -- in a local memory cell.
Each process can access its own value, but it cannot read the contents of another process (no data
sharing).

Term $\eta;C$ is an interaction between two processes, read
``the system may execute $\eta$ and proceed as $C$''.
An interaction $\eta$ is either a value communication -- $\gencom$ -- or a label
selection -- $\gensel$.
In $\gencom$, $\pid p$ sends its local
evaluation of expression $e$ to $\pid q$, which stores the
received value.
Expressions are either the constant $\zero$, the
value of the sender (written as $\pcont$), or an
application of the successor operator to $\pcont$.
In $\gensel$, $\pid p$ communicates label $l$ (either $\lleft$ or $\lright$) to
$\pid q$.
\newtext
Labels are used to communicate decisions about control flow, rather than data (see 
\S~\ref{sec:selections}).
\oldtext
In a conditional $\gencond$, $\pid q$ sends its value to $\pid p$,
which checks if the received value is
equal to its own; the choreography proceeds as $C_1$, if
that is the case, or as $C_2$, otherwise.
In value communications,
selections and conditionals, the two interacting processes must be different (no
self-communications).
Definitions and invocations of recursive
procedures are standard.
The term $\nil$, also called \emph{exit point}, is the terminated
choreography.

\newtext
In the remainder, we write $\pn(C)$ for the set of all process names that appear in $C$.

\oldtext

\subsection{Semantics of CC}
The semantics of CC uses reductions of the form
$C,\sigma \to C',\sigma'$.
The total state function $\sigma$ maps each process
name to its value.
We use $v$, $w$, $\ldots$ to range over values:
$v,w,\ldots ::= \zero \mid \suc v$.
Values are isomorphic to natural numbers via
$\numeral{n}=\s^n\cdot\zero$.
The reduction relation $\to$ is defined by the rules given in Figure~\ref{fig:cc_semantics}.
%
\input{cc_semantics}

Rule~\rname{C}{Com} models a value communication $\gencom$.
In the premise, we write $e[\sigma(\pid p)/\pcont]$ for the result of replacing $\pcont$ with 
$\sigma(\pid p)$ in $e$. In the reductum, $\sigma[\pid q \mapsto v]$ denotes the updated state function $\sigma$ where 
$\pid q$ now maps to $v$.
\newtext
Rule~\rname{C}{Sel}
states that selections are no-ops from the choreography's point of view (they do not 
change the state) -- see \S~\ref{sec:selections} for a more detailed explanation of why selections 
are useful.
\oldtext
Rule~\rname{C}{Ctx} is standard.

%
\input{cc_precongr}

\newtext
Rule~\rname{C}{Struct} closes the reduction relation under the \emph{structural precongruence} 
$\precongr$, which is the smallest precongruence satisfying the rules
in Figure~\ref{fig:cc_precongr}.
%
\oldtext
We write $C \equiv C'$ for $C \precongr C'$ and $C' \precongr C$.

\newtext
Structural precongruence allows non-interfering actions to be executed in 
any order, modelling concurrent process execution. This is achieved by rules \rname{C}{Eta-Eta}, 
\rname{C}{Eta-Cond} and \rname{C}{Cond-Cond}, which swap two terms describing actions performed by 
independent processes.
They observe the same conditions as in previous choreography models, where these rules are used for 
the swapping relation $\swapC$, e.g.,
in~\cite{CM13}. For example, rule~\rname{C}{Eta-Eta} allows swapping of two interactions $\eta$ and 
$\eta'$ -- $\eta;\eta' \equiv \eta';\eta$ -- whenever the processes that enact 
them are distinct -- $\pn(\eta) \cap \pn(\eta') = \emptyset$.

Rules~\rname{C}{ProcEnd} and \rname{C}{Unfold} are standard, and respectively deal with garbage 
collection of procedure definitions and unfolding of recursive procedures.

In rule~\rname{C}{Unfold}, we use the standard notion of context, denoted $C[]$. We 
borrow its definition (adapted to our language) from~\cite{SW01}.
A context $C[]$ is obtained when the hole $\bullet$ (a new reserved term) replaces one occurrence of 
$\nil$ in a choreography term given by the grammar in Figure~\ref{fig:cc_syntax}.
We write $C[C']$ for the choreography obtained by replacing the hole $\bullet$ in $C[]$ with $C'$.
Therefore, in rule~\rname{C}{Unfold}, we are identifying a specific subterm $X$ in the choreography 
on the left ($C_1[X]$) and then replacing it with the body of the recursive procedure on the right 
($C_1[C_2]$).
\oldtext

\subsection{Label Selection and Minimal Choreographies}
\label{sec:selections}

To the reader unfamiliar with choreographies, the role of selection -- $\gensel$ -- may be unclear 
at this point.
In existing choreography calculi, they are crucial in making choreographies projectable, as we illustrate with
an example.

\begin{example}
\label{ex:selection}
Consider the following choreography.
\[
C = \cond{\eqcom{\pid p}{\pid q}}{
	(\com{\pid p.\pcont}{\pid r};\nil)
}{
	(\com{\pid r.\pcont}{\pid p}; \nil)
  }
\]
Here, $\pid p$ checks whether its value is the same as that of $\pid q$. If so, $\pid p$
communicates its value to $\pid r$; otherwise, it is $\pid r$ that communicates its value to $\pid p$.
Recall that processes are assumed to run independently and share no data. Here, $\pid p$ is the only
process that knows which branch of the conditional should be executed.
However, $\pid r$ also needs to know this information, since it must behave differently.
Intuitively, we have a problem because we are asking
$\pid r$ to act differently based on a decision made by another process, $\pid p$, and there is no propagation of
this decision from $\pid p$ to $\pid r$ (either directly or indirectly, through other processes).
We can easily fix the example by adding selections:
\[
C' = \cond{\eqcom{\pid p}{\pid q}}{
  (\sel{\pid p}{\pid r}{\lleft}; \com{\pid p.\pcont}{\pid r};\nil)
}{
  (\sel{\pid p}{\pid r}{\lright}; \com{\pid r.\pcont}{\pid p}; \nil)
  }
\,.\]
Now, $\pid p$ tells $\pid r$ about its choice by sending a different label.
This intuition will be formalised in our definition of
EndPoint Projection in~\S~\ref{sec:epp}.
The choreography $C$ (without label selections) is not projectable, whereas $C'$ is.
\qed
\end{example}

\newtext
The example illustrates the role of label selections in choreography languages:
they are not meant to carry data, but rather to select 
one from several possible different behaviours offered by the process receiving the label.
Thus, they model interface selections, as in the invocation of a particular method in 
object-oriented programming, or an operation in service-oriented computing.
\oldtext

One of the main results in this paper is that
the same effect can actually be achieved in a language without label selections.
\newtext
However, this hides the way in which control flow is propagated from one process to another.
Furthermore, label selections are a typical ingredient of choreography calculi, and it makes sense 
to include them in a core model. This is also reflected in typical implementations of 
endpoint languages, where label selections and value communications are implemented as different 
primitives. See \S~\ref{sec:related} for a more thorough discussion.
\oldtext

In this work, we also consider the fragment of CC without label selections,
which we call Minimal Choreographies (MC).
The syntax of MC is the same as that of CC, except that the action $\gensel$ is disallowed; the semantics of
MC is the same as that of CC, disregarding the rules that pertain to terms outside the language.

\subsection{Properties}

In the remainder of this section we discuss some properties of Core Choreographies.
Unless otherwise stated, those properties also hold for Minimal Choreographies.

CC enjoys the usual deadlock-freedom-by-design property of
choreographies.
\begin{theorem}[Deadlock-freedom by design]
  \label{thm:df-by-design}
  If $C$ is a choreography, then either:
  \begin{itemize}
  \item $C \precongr \nil$ ($C$ has terminated);
  \item or, for all $\sigma$, $C,\sigma \to C',\sigma'$ for some $C'$ and $\sigma'$ ($C$ can reduce).
  \end{itemize}
\end{theorem}
\begin{proof}
  Assume that $C\not\precongr\nil$ (otherwise the result is trivial).
  The proof is by structural induction on $C$.
  If $C$ is not of the form $\genrec$, then the thesis is a consequence of the semantics, since there is always an applicable reduction rule.
  If $C$ is of the form $\genrec$, then also $C_1\not\preceq\nil$, and the thesis follows by 
induction hypothesis applied to $C_1$.
\end{proof}

Using this theorem,
in \S~\ref{sec:epp} we prove that the process implementations obtained by projecting choreographies is deadlock-free.

The semantics of CC suggests a natural definition of computation.
We write $\to^*$ for the transitive closure of $\to$ and $C,\sigma \not\to^\ast \nil$ for
$C,\sigma \not\to^* \nil,\sigma'$ for any~$\sigma'$.
\begin{definition}
  \label{defn:implementation}
  A choreography $C$ \emph{implements} a function
  $f:\NN^n\to\NN$ with input processes $\pid p_1,\ldots,\pid p_n$ and
  output process $\pid q$ if, for all $x_1,\ldots,x_n\in\NN$ and for
  every state $\sigma$ s.t. $\sigma(\pid p_i)=\numeral{x_i}$:
  \begin{itemize}
  \item if $f(\til x)$ is defined, then
    $C,\sigma\to^\ast\nil,\sigma'$ where
    $\sigma'(\pid q)=\numeral{f(\til x)}$;
  \item if $f(\til x)$ is undefined, then $C,\sigma\not\to^\ast\nil$.
  \end{itemize}
\end{definition}
By Theorem~\ref{thm:df-by-design}, in the second case $C,\sigma$ must
reduce infinitely (diverge).

In later sections, we need to characterize choreographies that are equivalent wrt a set of processes.
We use the state function $\sigma$ for this purpose.
\begin{definition}[Computational equivalence]
  Two states $\sigma_1,\sigma_2$ are \emph{equivalent} wrt a set of process names $\pids p$, written
  $\sigma_1\equiv_{\pids p}\sigma_2$, if $\sigma_1(\pid p)=\sigma_2(\pid p)$ for every $\pid p\in\pids p$.

  Two choreographies $C_1$ and $C_2$ are \emph{equivalent} wrt a set of process names $\pids p$ if: whenever
  $\sigma_1\equiv_{\pids p}\sigma_2$, if $C_1,\sigma_1\to^\ast\nil,\sigma'_1$ then $C_2,\sigma_2\to^\ast\nil,\sigma'_2$
  with $\sigma'_1\equiv_{\pids p}\sigma'_2$, and conversely.
\end{definition}

Throughout this paper, we focus on choreographies with only one exit point (a single occurence of $\nil$).
When $C$ has a single exit point, we write $C \gseq C'$ for the choreography obtained by replacing
$\nil$ in $C$ with $C'$.
(Requiring $C$ to have a single exit point makes this construction linear in the sizes of $C$ and $C'$, and
simplifies its theoretical analysis.)
This does not add expressivity to CC, but it allows for the usage of
macros (as in the examples below).
Then, $C\gseq C'$ behaves as a ``sequential composition'' of $C$ and $C'$, as induction over $C$ shows.
\begin{lemma}
  \label{lem:seqcomp}
  Let $C$ have one exit point, $C'$ be a choreography, $\sigma,\sigma',\sigma''$ be states.
  \begin{enumerate}
  \item If $C,\sigma\to^\ast\nil,\sigma'$ and
    $C',\sigma'\to^\ast\nil,\sigma''$, then
    $C\gseq C',\sigma\to^\ast\nil,\sigma''$.
  \item If $C,\sigma\not\to^\ast\nil$, then
    $C\gseq C',\sigma\not\to^\ast\nil$.
  \item If $C,\sigma\to^\ast\nil,\sigma'$ and
    $C',\sigma'\not\to^\ast\nil$, then
    $C\gseq C',\sigma\not\to^\ast\nil$.
  \end{enumerate}
\end{lemma}
\begin{proof}
  Straightforward by structural induction on $C$.
\end{proof}

Structural precongruence gives $C\gseq C'$ fully parallel behaviour in some cases.
Intuitively, $C_1$ and $C_2$ run in parallel in $C_1\gseq C_2$ if
their reduction paths to $\nil$ can be interleaved in any possible
way. Below, we write $C \lmto{\til\sigma} \nil$ for $C,\sigma_1\to C_2,\sigma_2 \to \cdots \to \nil,\sigma_n$,
where $\til\sigma=\sigma_1,\ldots,\sigma_n$, and $\wtil{\sigma(\pid p)}$ for the sequence
$\sigma_1(\pid p),\ldots,\sigma_n(\pid p)$.

\begin{definition}
  Let $\pids p$ and $\pids q$ be disjoint.
  Then, $\til\sigma$ is an \emph{interleaving} of $\wtil{\sigma_1}$ and $\wtil{\sigma_2}$ wrt $\pids p$ and $\pids q$
  if $\til\sigma$ 
  contains two subsequences
  $\wtil{\sigma'_1}$ and $\wtil{\sigma'_2}$ such that:
  \begin{itemize}
  \item $\wtil{\sigma'_2} = \wtil\sigma \setminus \wtil{\sigma'_1}$;
  \item $\wtil{\sigma'_1(\pid p)}=\wtil{\sigma_1(\pid p)}$ for all
    $\pid p\in\pids p$, and
    $\wtil{\sigma'_2(\pid q)}=\wtil{\sigma_2(\pid q)}$ for all
    $\pid q\in\pids q$;
  \item $\wtil{\sigma(\pid r)}$ is a constant sequence for all
    $\pid r\not\in\pids p\cup\pids q$.
  \end{itemize}
\end{definition}

\begin{definition}[Parallel Run]
\label{defn:parrun}
  Let $C_1$ and $C_2$ be choreographies such that
  $\pn(C_1)\cap\pn(C_2)=\emptyset$ and $C_1$ has only one exit point.
  We say that $C_1$ and $C_2$ \emph{run in parallel} in $C_1\gseq C_2$
  if:
  whenever $C_i\lmto{\til{\sigma_i}}\nil$, then
  $C_1\gseq C_2\lmto{\til\sigma}\nil$ for every interleaving
  $\til\sigma$ of $\til{\sigma_1}$ and $\til{\sigma_2}$ wrt $\pn(C_1)$ and $\pn(C_2)$.
\end{definition}

\begin{theorem}[Parallelisation]
  \label{thm:par}
  Let $C_1$ and $C_2$ be choreographies such that
  $\pn(C_1)\cap\pn(C_2)=\emptyset$ and $C_1$ has only one exit point.
  Then $C_1$ and $C_2$ run in parallel in $C_1\gseq C_2$.
\end{theorem}
\begin{proof}
  The result follows directly by induction over $C_1$.
\end{proof}

Definition~\ref{defn:parrun} and Theorem~\ref{thm:par} straightforwardly generalise to an arbitrary number of processes.
We provide an example of such parallel behaviour in Theorem~\ref{thm:par-compute}.

\subsection{Examples}

We present examples of choreographies in CC, writing them as macros (syntax shortcuts).
We use the notation $\macro{m}( \textit{params} ) \defeq C$, where $\macro{m}$ is the name of the macro,
$\textit{params}$ its parameters, and $C$ its body.

\begin{example}
\label{ex:inc_add}
The macro $\macro{inc}(\pid p, \pid t)$ increments the value of $\pid p$ using an auxiliary process $\pid t$.
\[
\macro{inc}(\pid p, \pid t) \quad \defeq\quad \com{\pid p.\pcont}{\pid t};\ \com{\pid t.(\suc \pcont)}{\pid p};\ \nil
\]

Using $\macro{inc}$, we write a macro $\macro{add}(\pid p, \pid q, \pid r, \pid t_1, \pid t_2)$ that adds the
values of $\pid p$ and $\pid q$ and stores the result in $\pid p$, using auxiliary processes $\pid r$, $\pid t_1$
and $\pid t_2$.
We follow the intuition as in low-level abstract register machines.
First, $\pid t_1$ sets the value of $\pid r$ to zero, and then calls procedure $X$, which increments the value of
$\pid p$ as many times as the value in $\pid q$.
In the body of $X$, $\pid r$ checks whether its value is the same as $\pid q$'s.
If so, it informs the other processes that the recursion will terminate (selection of $\lleft$); otherwise, it asks
them to do another step (selection of $\lright$).
In each step, the values of $\pid p$ and $\pid r$ are incremented using $\pid t_1$ and $\pid t_2$
as auxiliary processes.
The compositional usage of $\macro{inc}$ is allowed, as it has exactly one exit point.
\begin{align*}
\lefteqn{\macro{add}(\pid p, \pid q, \pid r, \pid t_1, \pid t_2) \defeq}\qquad\\
&\m{def} X =
  \m{if}\,{\eqcom{\pid r}{\pid q}}\,
  \m{then}\,
    \sel{\pid r}{\pid p}{\lleft};
    \sel{\pid r}{\pid q}{\lleft};
    \sel{\pid r}{\pid t_1}{\lleft};
    \sel{\pid r}{\pid t_2}{\lleft};
    \nil
  \\
  &\hspace{3.5em}
    \m{else}\,
    \sel{\pid r}{\pid p}{\lright};
    \sel{\pid r}{\pid q}{\lright};
    \sel{\pid r}{\pid t_1}{\lright};
    \sel{\pid r}{\pid t_2}{\lright};
    \macro{inc}(\pid p,\pid t_1)\gseq
    \macro{inc}(\pid r,\pid t_2)\gseq
    \call X\\
  &\m{in}\,\com{\pid t_1.\zero}{\pid r}; \call X
\end{align*}
By Theorem~\ref{thm:par}, the calls to $\macro{inc}(\pid p,\pid{t_1})$
and $\macro{inc}(\pid r,\pid{t_2})$ can be executed in parallel.
Indeed, applying rule $\rname{C}{Eta-Eta}$ for $\precongr$ repeatedly we can check that:
\begin{multline*}
\underbrace{\com{\pid p.\pcont}{\pid t_1};\ \com{\pid t_1.(\suc \pcont)}{\pid p};}_{\mbox{\footnotesize expansion of $\macro{inc}(\pid p,\pid t_1)$}}
\underbrace{\com{\pid r.\pcont}{\pid t_2};\ \com{\pid t_2.(\suc \pcont)}{\pid r};}_{\mbox{\footnotesize expansion of $\macro{inc}(\pid r,\pid t_2)$}}
X
\\ \precongr \
\underbrace{\com{\pid r.\pcont}{\pid t_2};\,\com{\pid t_2.(\suc \pcont)}{\pid r};}_{\mbox{\footnotesize expansion of $\macro{inc}(\pid r,\pid t_2)$}}
\underbrace{\com{\pid p.\pcont}{\pid t_1};\,\com{\pid t_1.(\suc \pcont)}{\pid p};}_{\mbox{\footnotesize expansion of $\macro{inc}(\pid p,\pid t_1)$}}
X
\end{multline*}
\qed
\end{example}


%% file: cc_syntax.tex
\begin{figure}
\begin{align*}
 C & ::= \eta; C \acspar \gencond 
\acspar  \genrec \acspar \gencall \acspar \nil
\\
\eta & ::= \gencom \acspar \gensel
\qquad
e ::= \zero \  | \ \pcont \ | \ \suc \pcont
\qquad
l ::= \lleft \ |\  \lright
\end{align*}
\caption{Core Choreographies, Syntax.}
\label{fig:cc_syntax}
\end{figure}

%% file: cc_semantics.tex
\begin{figure}
\begin{eqnarray*}
&\infer[\rname{C}{Com}]
{
	\gencom;C,\sigma
	\to
	C, \sigma[\pid q \mapsto v]
}
{
	v = e[\sigma(\pid p)/\pcont]
}
\qquad
\infer[\rname{C}{Sel}]
{
	\gensel;C, \sigma \to C, \sigma
}
{}
\\[1ex]
&\infer[\rname{C}{Cond}]
{
	\gencond, \sigma  \to    C_i, \sigma
}
{
	i = 1 \ \text{if } \sigma(\pid p) = \sigma(\pid q),\ 
	i = 2 \ \text{o.w.}
}
\\[1ex]
&
\infer[\rname{C}{Ctx}]
{
	\genrec, \sigma \ \to \ 
	\rec{X}{C_2}{C'_1}, \sigma'
}
{
	C_1, \sigma \ \to \   C'_1, \sigma'
}
\\[1ex]
&\infer[\rname{C}{Struct}]
{
	C_1, \sigma \to  C'_1, \sigma'
}
{
	C_1 \precongr C_2
	& C_2, \sigma \to   C'_2, \sigma'
	& C'_2  \precongr C'_1
}
\end{eqnarray*}
\caption{Core Choreographies, Semantics.}
\label{fig:cc_semantics}
\end{figure}

%% file: cc_precongr.tex
\begin{figure}
\begin{eqnarray*}
&
\infer[\rname{C}{Eta-Eta}]
{
	\eta;\eta'\ \equiv \ \eta';\eta
}
{
	\pn(\eta) \cap \pn(\eta') = \emptyset
}
\qquad
\infer[\rname{C}{ProcEnd}]
{
	\rec{X}{C}{\nil} \ \, \precongr \ \, \nil
}
{}
\\[1ex]
&
\infer[\rname{C}{Eta-Cond}]
{
	\cond{\eqcom{\pid p}{\pid q}}{(\eta;C_1)}{(\eta;C_2)}
	\quad \equiv\quad 
	\eta;\gencond
}
{
	\{ \pid p, \pid q\} \cap \pn(\eta) = \emptyset
}
\\[1ex]
&
\infer[\rname{C}{Eta-Rec}]
{
	\rec{X}{C_2}{(\eta;C_1)}
	\quad \equiv\quad 
	\eta;\genrec
}
{
	\pn(C_i) \cap \pn(\eta) = \emptyset
}
\\[1ex]
&\infer[\rname{C}{Cond-Cond}]
{
	\begin{array}{c}
	\cond{\eqcom{\pid p}{\pid q}}{
 		(\cond{\eqcom{\pid r}{\pid s}}
 		{C_1}{C_2})
	}{
 		(\cond{\eqcom{\pid r}{\pid s}}
 		{C'_1}{C'_2})
	}
	\\
	\ \equiv \
	\\
	\cond{\eqcom{\pid r}{\pid s}}{
		(\cond{\eqcom{\pid p}{\pid q}}
		{C_1}{C'_1})
	}{
		(\cond{\eqcom{\pid p}{\pid q}}
		{C_2}{C'_2})
	}
	\end{array}
}
{
	\{ \pid p, \pid q \} \cap \{ \pid r, \pid s \} = \emptyset
}
\\[1ex]
&
\infer[\rname{C}{Unfold}] 
{
	\rec{X}{C_2}{C_1[\gencall]}
\quad \precongr\quad
	\rec{X}{C_2}{C_1[C_2]}
}
{}
\end{eqnarray*}
\caption{Core Choreographies, Structural precongruence $\precongr$.}
\label{fig:cc_precongr}
\end{figure}

%% file: sp.tex
We present Stateful Processes (SP), our target process model, and define an EndPoint Projection (EPP) that
synthesises process implementations from choreographies in CC.
By restricting SP adequately, we obtain a target process calculus for MC, which we call Minimal Processes
(MP), with a corresponding, simpler, EPP.

\subsection{Syntax of SP and MP}
The syntax of SP is given in 
Figure~\ref{fig:sp_syntax}.
%
\input{sp_syntax}

Networks ($N,M$) are either the inactive network
$\emptyN$ or parallel compositions of processes
$\actor{\pid p}{v}{B}$, where $\pid p$ is the name of the process, $v$
its stored value, and $B$ its behaviour.

We comment on behaviours.
Expressions and labels are as in CC.
A send term $\asend{\pid q}{e};B$ sends the evaluation of expression $e$ to $\pid q$, proceeding as $B$.
Term $\arecv{\pid p};B$, the dual receiving action, stores the value received from $\pid p$ in the process
executing the behaviour, proceeding as $B$.
A selection term $\asel{\pid q}{l};B$ sends $l$ to $\pid q$.
Dually, a branching term $\abranch{\pid p}{\{ l_i : B_i\}_{i\in I}}$ receives one of the labels $l_i$ and proceeds
as $B_i$.
A process offers either: a single branch (labeled $\lleft$ or $\lright$); or two branches (with distinct labels).
In a conditional $\cond{\eqcom{\pcont}{\pid q}}{B_1}{B_2}$, the process receives a value from process $\pid q$ and
compares it with its own value to choose the continuation $B_1$ or $B_2$.
The other terms (definition/invocation of recursive procedures, termination) are standard.

The syntax of MP is obtained by disallowing the terms $\asel{\pid q}{l};B$ and
$\abranch{\pid p}{\{ l_i : B_i\}_{i\in I}}$.

\subsection{Semantics}

The reduction rules for SP are mostly standard, from process calculi, and are included in Figure~\ref{fig:sp_semantics}.
The key difference from CC is that execution is now distributed over processes.
\input{sp_semantics}
Rule \rname{S}{Com} follows the standard communication rule in process calculi.
A process $\pid p$ executing a send action towards a process $\pid q$ can synchronise with a receive-from-$\pid p$
action at $\pid q$; in the reduct, $\pid q$'s value is updated with the value sent by $\pid p$, obtained by
replacing the placeholder $\pcont$ in $e$ with the value of $\pid p$.
Rule \rname{S}{Sel} is selection from session types~\cite{HVK98}, with the sender selecting one of the branches
offered by the receiver.
In rule \rname{S}{Cond}, $\pid p$ (executing the conditional) acts as a receiver for the value sent by the process
whose value it wants to read ($\pid q$).
\newtext
All other rules are standard, and use a structural precongruence that includes associativity and commutativity of parallel composition, together with the rules in Figure~\ref{fig:sp_precongr}, supporting
recursion unfolding
and garbage collection of terminated processes and unused definitions.
\oldtext
\input{sp_precongr}

As for CC, we can define function implementation in SP.

\begin{definition}[Function implementation in SP]
  A network $N$ \emph{implements} a function $f:\NN^n\to\NN$
  with input processes $\pid p_1,\ldots,\pid p_n$ and output process
  $\pid q$ if $N \precongr (\prod_{i\in[1,n]}\actor{\pid p_i}{v_i}{B_i}) \parp\actor{\pid q}{w}{B'}\parp N'$
  and, for all $x_1,\ldots,x_n\in\NN$:
  \begin{itemize}
  \item if $f(\til x)$ is defined, then $N(\til x)\to^\ast\actor{\pid q}{\numeral{f(\til x)}}{\nil}$;
  \item if $f(\til x)$ is not defined, then $N(\til x)\not\to^\ast\nil$.
  \end{itemize}
  where $N(\til x)$ is a shorthand for $N[\wtil{\numeral{x_i}/v_i}]$, the network obtained by replacing in $N$ the
  values of the input processes with the arguments of the function.
\end{definition}

\subsection{EndPoint Projection}
\label{sec:epp}
We now define an EndPoint Projection (EPP) from CC to SP.

We first discuss the rules for projecting the behaviour of a single
process $\pid p$, a partial function $\epp{C}{\pid p}$ defined
by the rules in Figure~\ref{fig:epp}.
Selections are projected similarly to communications, and $\nil$ is projected to $\nil$.
\input{epp}
All rules follow the intuition of projecting, for each choreography
term, the local action performed by the process that we are
projecting.
For example, for a communication term $\gencom$, we project a send
action for the sender $\pid p$, a receive
action for the receiver $\pid q$, or just
the continuation otherwise.
The rule for selection is similar.

\newtext
The rules for projecting recursive definitions and recursive calls are defined assuming
that procedure names have been annotated with the process names appearing
inside the body of the procedure, in order to avoid projecting
unnecessary procedure code -- see~\cite{CHY12}.
This is an easy preprocessing of the choreography, which we assume is performed before using 
projection. The preprocessing recursively visits the choreography and leaves 
it untouched, except when it meets a recursive definition term $\rec{X}{C_2}{C_1}$. In 
that case, it annotates the choreography -- $\rec{X^{\pids p}}{C_2\{X^{\pids p}/X\}}{C_1\{X^{\pids 
p}/X\}}$ where $\{\pids p\} = \pn(C_2)$ -- and then proceeds recursively. In other words, the 
definition of $X$ and all its invocations in $C_2$ and $C_1$ are annotated with ${\pids p}$.
\oldtext

The rule for projecting a conditional is more involved, using the partial merging operator $\merge$ to merge
the possible behaviours of a process that does not know which branch will be chosen.
The formal definition is found in Figure~\ref{fig:epp_merge}.
\input{epp_merge}
Merging is a homomorphic binary operator; for all terms but branchings it requires isomorphism, 
$\asend{\pid q}{e};B \merge \asend{\pid q}{e};B' = \asend{\pid q}{e};(B \merge B')$.
The only case where branching terms can have unmergeable continuations is when they are guarded by distinct
labels, in which case merge returns a larger branching including all options (merging branches with the same
label).

Merging explains the role of selections in CC, common in choreography
models~\cite{BCDLDY08,CHY12,CM13,HYC08,DGGLM17,QZCY07}.
Recall the choreographies from Example~\ref{ex:selection}.
In choreography $C$, the behaviour of $\pid r$ cannot be projected because we cannot merge its different behaviours in
the two branches of the conditional (a send with a receive).
Choreography $C'$ is projectable, and the behaviour of $\pid r$ is
$\epp{C}{\pid r} = \abranch{\pid p}{\{\lleft:\arecv{\pid p}; \nil,\ \lright:\asend{\pid p}{\pcont};\nil\}}$.

\begin{definition}[EPP from CC to SP]
Given a choreography $C$ and a state $\sigma$, the \emph{endpoint projection} of $C$ and $\sigma$ is the parallel
composition of the projections of the processes in $C$.
\[
\epp{C,\sigma}{}  =
\prod_{\pid p \in \pn(C)} \actor{\pid p}{\sigma(\pid p)}{\epp{C}{\pid p}}
\]
\end{definition}
The EPP from MC to MP is defined by restricting the EPP from CC to SP to the relevant cases.
Since choreographies in MC do not have selections, process projections of choreographies in MC never have branchings.
This means that, in the case of MC, the merging operator $\sqcup$ used in EPP is exactly syntactic equality (since the 
only nontrivial case was that of branchings).
Consequently, we can replace the rule for projecting conditionals with the following, simpler, one.

\[
\epp{\gencond}{\pid r} =
	\begin{cases}
		\cond{\eqcom{\pcont}{\pid q}}{\epp{C_1}{\pid r}}{\epp{C_2}{\pid r}} & \text{if } \pid r = \pid p \\
		\asend{\pid p}{\pcont}; \epp{C_1}{\pid r} & \text{if } \pid r = \pid q \text{ and } \epp{C_1}{\pid r} = 
\epp{C_2}{\pid r} \\
		\epp{C_1}{\pid r} & \mbox{if } \pid r \not\in \{\pid p, \pid q\} \text{ and } \epp{C_1}{\pid r} = 
\epp{C_2}{\pid r}
	\end{cases}
\]

Since the state function $\sigma$ is total, $\epp{C,\sigma}{}$ is defined for some $\sigma$ iff
$\epp{C,\sigma'}{}$ is defined for all other $\sigma'$.
In this case, we say that $C$ is \emph{projectable}.

\begin{example}
Recall the definition of $\macro{inc}(\pid p, \pid t)$.
\[
\macro{inc}(\pid p, \pid t) \defeq \com{\pid p.\pcont}{\pid t};\ \com{\pid t.(\suc \pcont)}{\pid p};\ \nil
\]
Given any $\sigma$, the EPP of $\macro{inc}(\pid p, \pid t)$ is:
\[
\epp{\macro{inc}(\pid p, \pid t),\sigma}{} =
\ \actor{\pid p}{\sigma(\pid p)}{\asend{\pid t}{\pcont}; \arecv{\pid t}; \nil}
\ \parp\
\actor{\pid t}{\sigma(\pid t)}{\arecv{\pid p}; \asend{\pid p}{\suc\pcont}; \nil}
\]
\end{example}

\newtext
In order to characterize the correspondence between choreographies and their projections, we introduce a \emph{pruning relation} between networks, inspired by~\cite{CHY12,CM13}.

\begin{definition}
  The pruning relation $\sqsupseteq$ on networks is defined as $\nil\sqsupseteq\nil$ and 
$\actor{\pid p}vB\parp N\sqsupseteq\actor{\pid p}v{B'}\parp N'$ if $B\merge B'=B$ and $N\sqsupseteq 
N'$.
\end{definition}
Intuitively, $N\sqsupseteq N'$ if both $N$ and $N'$ are parallel compositions of the same processes, 
and each process in $N$ offers at least the same behaviours as in $N'$, with potentially 
extra options in its branching terms.

Using pruning, we can formalise the operational correspondence guaranteed by EPP.


\begin{theorem}[Operational Correspondence (CC $\leftrightarrow$ SP)]
\label{thm:oc-cc-sp}
Let $C$ be a projectable choreography. Then, for all $\sigma$:
\begin{description}
\item[Completeness:] If $C,\sigma \to C',\sigma'$, then $\epp{C,\sigma}{} \to \sqsupseteq \preceq 
\epp{C',\sigma'}{}$;
\item[Soundness:] If $\epp{C,\sigma}{} \to N$, then there exist $C'$, $\sigma'$ and $N'$ such that
  $C,\sigma \to C',\sigma'$ and $N \precongr N' \sqsupseteq \epp{C',\sigma'}{}$.
\end{description}
\end{theorem}
\begin{proof}
  The proof is an adaptation of the proofs of similar results in~\cite{CHY12,CM13}; we sketch a few typical cases, also to illustrate where pruning plays a role.
  \begin{description}
  \item[Completeness:] by induction on the derivation of the reduction $C,\sigma\to C',\sigma'$.
    \begin{itemize}
    \item Consider the case of rule \rname{C}{Com}. Then $C=\gencom;C'$, $\sigma'=\sigma[\pid q\mapsto e[\sigma(\pid p)/\pcont]]$ and
      \begin{align*}
        \epp{C,\sigma}{}
        &=\actor{\pid p}{\sigma(\pid p)}{\asend{\pid q}{e};B_{\pid p}}\parp\actor{\pid q}{\sigma(\pid q)}{\arecv{\pid p};B_{\pid q}}\parp N'' \\
        &\to\actor{\pid p}{\sigma(\pid p)}{B_{\pid p}}\parp\actor{\pid q}{e[\sigma(\pid p)/v]}{B_{\pid q}}\parp N''
      \end{align*}
      If $B_{\pid p},B_{\pid q}\not\preceq\nil$, then the result of the reduction is precisely $\epp{C',\sigma'}{}$. If one or both of these behaviours are $\nil$, then $\preceq$ needs to be applied to remove the corresponding processes from the network in order to obtain $\epp{C',\sigma'}{}$.
    \item Consider the case of rule \rname{C}{Cond} with $i=1$ (i.e.\ $\sigma(\pid p)=\sigma(\pid q)$). Then $C=\gencond$, $C'=C_1$, $\sigma'=\sigma$ and
      \begin{align*}
        \epp{C,\sigma}{}
        &=\actor{\pid p}{\sigma(\pid p)}{\cond{\eqcom{\pcont}{\pid q}}{B_1}{B_2}}\parp\actor{\pid q}{\sigma(\pid q)}{\asend{\pid p}{\pcont};B_{\pid q}}\parp N'' \\
        &\to\actor{\pid p}{\sigma(\pid p)}{B_1}\parp\actor{\pid q}{\sigma(\pid q)}{B_{\pid q}}\parp N''
      \end{align*}
      Assume first for simplicity that $B_1,B_{\pid q}\not\preceq\nil$.
      For each process $\pid r\neq\pid p$, the network obtained above is $\epp{C_1,\sigma}{\pid r}\merge\epp{C_2,\sigma}{\pid r}$, which in general does not coincide with $\epp{C_1,\sigma}{\pid r}$. However, since merge is trivially associative and idempotent, it immediately follows that $\actor{\pid p}{\sigma(\pid p)}{B_1}\parp\actor{\pid q}{\sigma(\pid q)}{B_{\pid q}}\parp N''\sqsupseteq\epp{C_1,\sigma}{}$.

      In the case that $B_1,B_{\pid q}\not\preceq\nil$, structural precongruence also needs to be applied to remove the respective processes from the resulting network.
    \item Finally consider the case when \rname{C}{Struct} applies, i.e.\ $C\preceq C_1$, $C_1,\sigma\to C'_1,\sigma'$ and $C'_1\preceq C'$.
      The thesis then follows by observing that rules \rname{C}{Eta-Eta}, \rname{C}{Eta-Cond}, \rname{C}{Eta-Rec} and \rname{C}{Cond-Cond} do not change the projected choreographies, while the behaviour of \rname{C}{Unfold} is reproducible by \rname{S}{Unfold}.
      Therefore $\epp{C,\sigma}{}\preceq\epp{C_1,\sigma}{}$ and the induction hypothesis applies.
      It remains to be shown that $\epp{C'_1,\sigma'}{}\preceq\epp{C',\sigma'}{}$, which also requires considering rule \rname{C}{ProcEnd} -- whose behaviour again can be directly replicated by \rname{S}{ProcEnd}.
    \end{itemize}
  \item[Soundness:] the proof is now by induction on the structure of $C$. We detail one representative case.
    \begin{itemize}
    \item Suppose that $C=\gencom;C'$. Then $\epp{C,\sigma}{}=\actor{\pid p}{\sigma(\pid p)}{\asend{\pid q}{e};B_{\pid p}}\parp\actor{\pid q}{\sigma(\pid q)}{\arecv{\pid p};B_{\pid q}}\parp N''$, and there are two cases.

      If $\epp{C,\sigma}{}\to N$ involves the communication between $\pid p$ and $\pid q$, then the choreography can mimic this action by consuming $\gencom$ and evolving to $C'$.
      Any applications of structural precongruence to $\epp{C,\sigma}{}$ must involve only rules \rname{S}{Unfold} and \rname{S}{ProcEnd}, which can be directly applied also to $C$ (and may require additional applications of the same rules to other processes in $N$, namely those that occur in the definition of the relevant variable $X$).
      Likewise, applications of structural precongruence to $N$ either involve the same rules (and the same argument applies to $\epp{C',\sigma}{}$) or garbage collection (which is transparent, since it can never be applied to $\epp{C',\sigma}{}$).
      Finally, additional garbage collection rules may be needed to obtain $\epp{C',\sigma'}{}$.
      In any case, $N\preceq\epp{C',\sigma'}{}$.

      If $\epp{C,\sigma}{}\to N$ does not involve $\pid p$ and $\pid q$, then we observe that necessarily $N$ has the form $\actor{\pid p}{\sigma(\pid p)}{\asend{\pid q}{e};B_{\pid p}}\parp\actor{\pid q}{\sigma(\pid q)}{\arecv{\pid p};B_{\pid q}}\parp N''$, where $\epp{C',\sigma}{}\to\actor{\pid p}{\sigma(\pid p)}{B_{\pid p}}\parp\actor{\pid q}{\sigma(\pid q)}{B_{\pid q}}\parp N''=N^-$.
      Then we can apply the induction hypothesis to $C'$, and find $C'',\sigma'$ such that $N^-\preceq N'\sqsupseteq\epp{C'',\sigma'}{}$ for some $N'$.
      Since the transition from $C'$ to $C''$ does not involve $\pid p$ or $\pid q$, a simple induction argument shows that $\gencom;C',\sigma\to\gencom;C'',\sigma'$, and from the definition of $\preceq$ and $\sqsupseteq$ it follows that $N\preceq N'''\sqsupseteq\epp{\gencom;C'',\sigma'}{}$ where $N'''$ is obtained from $N'$ by prepending the original communication actions between $\pid p$ and $\pid q$ at the head of these processes' behaviours.
    \end{itemize}
  \end{description}
\end{proof}
\oldtext
As a consequence of Theorems~\ref{thm:df-by-design} and~\ref{thm:oc-cc-sp},
choreography projections never deadlock.

\begin{theorem}[Deadlock-freedom by construction]
\label{cor:df-sp}
Let $N = \epp{C,\sigma}{}$ for some $C$ and $\sigma$. Then, either $N \precongr \emptyN$
($N$ has terminated), or $N \to N'$ for some $N'$ ($N$ can reduce).
\end{theorem}
\begin{proof}
  If $N\precongr\nil$ then the theorem clearly holds.
  Otherwise, the thesis follows from Theorems~\ref{thm:df-by-design} and~\ref{thm:oc-cc-sp}.
\end{proof}

\subsection{Choreography Amendment}
An important property of CC is that all unprojectable choreographies can be made projectable by adding some
selections.
We annotate recursion variables as for EPP, assuming that $\pn(X^{\pids p}) = \{\pids p\}$.

\begin{definition}[Amendment]
\label{def:amend}
Given $C$ in CC, the transformation $\amend(C)$ repeatedly applies the following procedure until no
longer possible, starting from the innermost subterms in $C$.
For each conditional subterm $\gencond$ in $C$, let $\pids r \subseteq (\pn(C_1)\cup\pn(C_2))$ be the largest set
such that
$\epp{C_1}{\pid r} \merge \epp{C_2}{\pid r}$ is undefined for all $\pid r \in \pids r$; then $\gencond$ in $C$ is replaced with:
\[
\m{if}\, (\eqcom{\pid p}{\pid q}) \, \m{then} \, (\sel{\pid p}{\pid r_1}{\lleft};\cdots;\sel{\pid p}{\pid r_n}{\lleft};C_1) \, \m{else} \,
 (\sel{\pid p}{\pid r_1}{\lright};\cdots;\sel{\pid p}{\pid r_n}{\lright};C_2)
\]
\end{definition}

\begin{lemma}[Amendment Lemma]
  \label{lem:amend}
  For every choreography $C$:
  \begin{description}
  \item[Completeness:] $\amend(C)$ is defined; 
  \item[Projectability:] for all $\sigma$, $\epp{\amend(C),\sigma}{}$ is defined;
  \item[Correspondence:] for all $\sigma$, $C,\sigma \to^* C',\sigma'$
    iff $\amend(C),\sigma \to^* \amend(C'),\sigma'$.
  \end{description}
\end{lemma}
\begin{proof}
  This procedure clearly terminates, since each iteration processes one conditional that is never changed again.
  By construction, the resulting choreography is projectable: if $\epp{C_1,\sigma}{\pid r}$ and $\epp{C_2,\sigma}{\pid r}$ are not mergeable, then the corresponding terms in $\epp{\amend(\gencond),\sigma}{\pid r}$ are both branching terms, making them mergeable.
  Correspondence is immediate, since the only new actions in $\amend(C')$ are label selections, which do not change the state when executed.
\end{proof}

\begin{example}
Applying $\amend$ to the choreography $C$ in Example~\ref{ex:selection} yields the choreography $C'$ in the
same example.
\qed
\end{example}

\begin{example}
Thanks to merging, amendment can also recognise some situations where additional selections are not needed.
For example, in the choreography
$C=\cond{\eqcom{\pid p}{\pid q}}{\left(\com{\pid p.(\suc \pcont)}{\pid r};\nil\right)}{\left(\com{\pid p.(\pcont)}{\pid r};\nil\right)}$,
$\pid r$ does not need to know the choice made by $\pid p$,
as it always performs the same input action.
Here, $C$ is projectable and $\amend(C) = C$.
\qed
\end{example}

Note that amending a choreography from MC returns a choreography in CC.
In \S~\ref{sec:selection} we show that we can prove an amendment lemma for MC, but this will require
much more work.


%% file: sp_syntax.tex
\begin{figure}
\begin{align*}
B ::={} & \asend{\pid q}{e};B \ \mid \ \arecv{\pid p};B \ \mid \ \asel{\pid q}{l};B \ \mid \abranch{\pid p}{\{ l_i : B_i\}_{i\in I}} \ \mid \\
\mid{} & \lefteqn{\nil
 \mid \cond{\eqcom{\pcont}{\pid q}}{B_1}{B_2} \ \mid \ \rec{X}{B_2}{B_1} \ \mid \ \call X} \\
N,M ::={} & \actor{\pid p}{v}{B} \ \mid\  \emptyN\ \mid\ N \parp M
\end{align*}
\caption{Stateful Processes, syntax.}
\label{fig:sp_syntax}
\end{figure}

%% file: sp_semantics.tex
\begin{figure}
\begin{eqnarray*}
&\infer[\rname{S}{Com}]
{
	\actor{\pid p}{v}{\asend{\pid q}{e};B_1}
	\ \parp\ 
	\actor{\pid q}{w}{\arecv{\pid p};B_2}
	\ \to \ 
	\actor{\pid p}{v}{B_1}
	\ \parp\ 
	\actor{\pid q}{u}{B_2}
}
{
	u = e[v/\pcont]
}
\\[1ex]
&\infer[\rname{S}{Sel}]
{
	\actor{\pid p}{v}{\asel{\pid q}{l_j};B}
	\ \parp\ 
	\actor{\pid q}{w}{\abranch{\pid p}{\{ l_i : B_i\}_{i\in I}}}
	\ \to \
	\actor{\pid p}{v}{B}
	\ \parp\ 
	\actor{\pid q}{w}{B_j}
}
{j \in I}
\\[1ex]
&
\infer[\rname{S}{Cond}]
{
	\actor{\pid p}{v}{\cond{\eqcom{\pcont}{\pid q}}{B_1}{B_2}}
	\ \parp\ 
	\actor{\pid q}{w}{\asend{\pid p}{e};B'}
	\ \to \ 
	\actor{\pid p}{v}{B_i}
	\ \parp \ 
	\actor{\pid q}{w}{B'}
}
{
	i = 1 \ \text{if } v = e[w/\pcont],\quad
	i = 2 \ \text{otherwise}
}
\\[1ex]
&
\infer[\rname{S}{Ctx}]
{
\actor{\pid p}{v}{\rec{X}{B_2}{B_1}}
\quad \to \quad
\actor{\pid p}{v}{\rec{X}{B_2}{B_1'}}
}{
  \actor{\pid p}{v}{B_1} \quad \to \quad \actor{\pid p}{v}{B_1'}
}
\\[1ex]
&\infer[\rname{S}{Par}]
{
	N \parp M \quad \to \quad N' \parp M
}
{
	N \ \to\  N'
}
\qquad
\infer[\rname{S}{Struct}]
{
	N \quad \to \quad N'
}
{
	N \precongr M & M\ \to\ M' & M' \precongr N'
}
\end{eqnarray*}
\caption{Stateful Processes, Semantics.}
\label{fig:sp_semantics}
\end{figure}

%% file: sp_precongr.tex
\begin{figure}
\centering
\begin{eqnarray*}
&\infer[\rname{S}{Unfold}]
{
	\rec{X}{B_2}{B_1[X]} \ \precongr \ \rec{X}{B_2}{B_1[B_2]}
}
{}
\\[1ex]
&
\infer[\rname{S}{PZero}]
{\actor{\pid p}{v}{\nil} \ \precongr \ \nil}
{}
\qquad
\infer[\rname{S}{NZero}]
{N \parp {\emptyN} \ \precongr \ N}
{}
\\[1ex]
&
\infer[\rname{S}{ProcEnd}]
{\rec{X}{B}{\nil}\ \precongr\ \nil }
{}
\end{eqnarray*}
\caption{Stateful Processes, Structural Precongruence.}
\label{fig:sp_precongr}
\end{figure}

%% file: epp.tex
\begin{figure*}
\begin{eqnarray*}
&
\epp{\gencom;C}{\pid r} =
	\begin{cases}
		\asend{\pid q}{e};\epp{C}{\pid r} & \text{if } \pid r = \pid p \\
		\arecv{\pid p};\epp{C}{\pid r} & \text{if } \pid r = \pid q \\
		\epp{C}{\pid r} & \mbox{o.w.}
	\end{cases}
\\[1ex]
&
\epp{\gensel;C}{\pid r} =
	\begin{cases}
		\asel{\pid q}{l};\epp{C}{\pid r} & \text{if } \pid r = \pid p \\
		\abranch{\pid p}{\{ l : \epp{C}{\pid r} \}} & \text{if } \pid r = \pid q \\
		\epp{C}{\pid r} & \mbox{o.w.}
	\end{cases}
\\[1ex]
&\epp{\gencond}{\pid r} =
	\begin{cases}
		\cond{\eqcom{\pcont}{\pid q}}{\epp{C_1}{\pid r}}{\epp{C_2}{\pid r}} & \text{if } \pid r = \pid p \\
		\asend{\pid p}{\pcont}; ( \epp{C_1}{\pid r} \merge \epp{C_2}{\pid r} ) & \text{if } \pid r = \pid q \\
		\epp{C_1}{\pid r} \merge \epp{C_2}{\pid r} & \mbox{o.w.}
	\end{cases}
\\[1ex]
&
	\epp{\rec{X^{\pids p}}{C_2}{C_1}}{\pid r} =
		\begin{cases}
			\rec{X}{\epp{C_2}{\pid r}}{\epp{C_1}{\pid r}}
			& \text{if } \pid r \in \pids p \\
			\epp{C_1}{\pid r} & \mbox{o.w.}
		\end{cases}
\\[1ex]
&\epp{\nil}{\pid r} = \nil
\qquad\qquad
\epp{X^{\pids p}}{\pid r} =
	\begin{cases}
		\call{X}
		& \text{if } \pid r \in \pids p \\
		\nil & \mbox{o.w.}
	\end{cases}
\end{eqnarray*}
\caption{Core Choreographies, Behaviour Projection.}
\label{fig:epp}
\end{figure*}

%% file: epp_merge.tex
\begin{figure}
\centering
\begin{eqnarray*}
\left(\asend{\pid q}{e};B\right) &\merge& \left(\asend{\pid q}{e};B'\right)
\quad = \quad \asend{\pid q}{e};(B \merge B')
\\[1ex]
\left(\arecv{\pid p};B\right) &\merge& \left(\arecv{\pid p};B'\right)
\quad = \quad \arecv{\pid p};(B\merge B')
\\[1ex]
\left(\asel{\pid q}{l};B\right)&\merge& \left(\asel{\pid q}{l};B'\right)
\quad = \quad \asel{\pid q}{l};(B \merge B')
\\[1ex]
\abranch{\pid p}{\{l_i:B_i\}_{i\in J}} &\merge&
\abranch{\pid p}{\{l_i:B'_i\}_{i\in K}} \quad = \quad \\
&& \hspace*{-2cm} \abranch{\pid p}{\left(\{l_i:(B_i \merge B'_i)\}_{i\in J\cap K}\cup\{l_i:B_i\}_{i\in J \setminus K}\cup\{l_i:B'_i\}_{i\in K \setminus J}\right)}
\\[1ex]
\left(\cond{\eqcom{\pcont}{\pid q}}{B_1}{B_2}\right) &\merge&
\left(\cond{\eqcom{\pcont}{\pid q}}{B'_1}{B'_2}\right) \quad = \quad \\
&& \hspace*{1cm}
\left(\cond{\eqcom{\pcont}{\pid q}}{(B_1 \merge B'_1)}{(B_2 \merge B'_2)}\right)
\\[1ex]
\call X&\merge& \call X \quad = \quad \call X
\\[1ex]
\left(\rec{X}{B_2}{B_1}\right) &\merge& \left(\rec{X}{B'_2}{B'_1}\right)
\quad = \quad \\
&& \hspace*{1cm} \left(\rec{X}{(B_2 \merge B'_2)}{(B_1 \merge B'_1)}\right)
\\[1ex]
B_1 &\merge& B_2 \quad = \quad B'_1 \merge B'_2 \quad \left(\text{if } B_1 \precongr B'_1 \text{ and } B_2 \precongr B'_2\right)
\end{eqnarray*}
\caption{Core Choreographies, Merge Operator in Behaviour Projection.}
\label{fig:epp_merge}
\end{figure}

%% file: kleene.tex
We now move to our main result: the set of choreography projections of CC (the processes synthesised by EPP) is
not only deadlock-free, but also capable of computing all
partial recursive functions, as defined by Kleene~\cite{Kleene52}, and hence Turing complete.
To this aim, the design and properties of CC give us a considerable pay off.
First, by Theorem~\ref{thm:oc-cc-sp}, the problem reduces to establishing that a projectable fragment of CC
is Turing complete. Second, by Lemma~\ref{lem:amend}, this simpler problem is reduced to establishing that
MC is Turing complete, since any choreography in MC can be amended to one in CC that is projectable and
computes the same values.
We also exploit the concurrent semantics of CC and Theorem~\ref{thm:par} to parallelise independent sub-computations
(Theorem~\ref{thm:par-compute}).
By projecting our choreographies via EPP, we obtain corresponding function implementations in the process calculus SP.

Establishing that CC is Turing complete is long, but not difficult.
Our proof is in line with other traditional proofs of computational completeness~\cite{Cutland80,Kleene52,Turing36},
where data and programs are distinct.
This differs from other proofs of similar results for, e.g.,
$\pi$-calculus~\cite{SW01} and $\lambda$-calculus~\cite{barendregt}, which encode data as
particular programs. The advantages are: our proof can be used to build choreographies that
compute particular functions; and we can parallelise independent sub-computations in functions (Theorem~\ref{thm:par-compute}).

\subsection{Partial Recursive Functions}
Our definition of the class of partial recursive functions $\RR$ is slightly simplified, but equivalent to, that
in~\cite{Kleene52}, where it is also shown that $\RR$ is the class of functions computable by a Turing machine.
$\RR$ is defined inductively as follows.

\begin{description}
\item[Unary zero:] $Z\in\RR$, where $Z:\NN\to\NN$ is s.t.~$Z(x)=0$ for all $x\in\NN$.
\item[Unary successor:] $S\in\RR$, where $S:\NN \to \NN$ is s.t.~$S(x)=x+1$ for all $x\in\NN$.
\item[Projections:] If $n\geq 1$ and $1\leq m\leq n$, then $P^n_m \in \RR$, where $P^n_m:\NN^n\to\NN$ is
  s.t.~$P^n_m(x_1,\ldots,x_n)=x_m$ for all $x_1,\ldots,x_n\in\NN$.
\item[Composition:] if $f,g_i\in\RR$ for $1\leq i\leq k$, with each $g_i:\NN^n\to\NN$ and $f:\NN^k\to\NN$, then
  $h=C(f,\til g)\in\RR$, where $h:\NN^n\to\NN$ is defined by composition from $f$ and $g_1,\ldots,g_k$ as:
  $h(\til x) = f(g_1(\til x),\ldots,g_k(\til x))$.
\item[Primitive recursion:] if $f,g\in\RR$, with $f:\NN^n\to\NN$ and $g:\NN^{n+2}\to\NN$, then $h=R(f,g)\in\RR$,
  where $h:\NN^{n+1}\to\NN$ is defined by primitive recursion from $f$ and $g$ as:
  $h(0,\til x)=f(\til x)$ and $h(x_0+1,\til x)=g(x_0,h(x_0,\til x),\til x)$.
\item[Minimization:] If $f\in\RR$, with $f:\NN^{n+1}\to\NN$, then $h=M(f)\in\RR$, where $h:\NN^n\to\NN$ is defined
  by minimization from $f$ as: $h(\til x)=y$ iff (1)~$f(\til x,y)=0$ and (2)~$f(\til x,y)$ is defined and different
  from $0$ for all $z<y$.
\end{description}

\begin{example}[Addition and Subtraction]
  \label{ex:sumsub}
  We show that $\Sum\in\RR$, where $\Sum:\NN^2 \to \NN$ adds its two arguments.
  Since $\Sum(0,y)=y$ and $\Sum(x+1,y)=\Sum(x,y)+1$, we can define $\Sum$ by recursion as
  $\Sum = R\big(P^1_1,C(S,P^3_2)\big)$. 
  Indeed, the function $y\mapsto\Sum(0,y)$ is simply $P^1_1$, whereas $1+\Sum(x,y)$ is $h(x,\Sum(x,y),y)$
  where $h(x,y,z)=y+1$, which is the composition of the successor function with $P^3_2$.

  From addition, we can define subtraction by minimization, since $\Sub(x,y)=x-y$ is the smallest $z$ such
  that $y+z = x$ (subtraction is not defined if $y>x$).
  We use an auxiliary function $\Eq(x,y)$ that returns $0$ if $x=y$ and a non-zero value otherwise, which is
  known to be partial recursive.
  Then we can define subtraction as $\Sub=M\Big(C\big(\Eq,C(\Sum,P^3_2,P^3_3),P^3_1\big)\Big)$.
  Indeed, composing $\Sum$ with $P^3_2$ and $P^3_3$ produces $(x,y,z)\mapsto y+z$, and the outer composition
  yields $(x,y,z)\mapsto\Eq(y+z,x)$.
  This function evaluates to $0$ precisely when $z=y-x$, and applying minimization computes this value from
  $x$ and $y$.
\end{example}

\subsection{Encoding Partial Recursive Functions in MC and CC}
\label{sec:encoding}
 
All functions in $\RR$ can be implemented in CC, in the sense of Definition~\ref{defn:implementation}.
Since selections can be inferred by amendment, we develop our encoding in MC and discuss projectability later.

Given $f:\NN^n\to\NN$, we denote its implementation
\newtext
in MC
\oldtext
by $\enc f{\pids p \mapsto\pid q}{}$, where $\pids p$ and $\pid
q$ are parameters.
All choreographies we build have a single exit point, and we combine them using the sequential 
composition operator
$\gseq$ from \S~\ref{sec:cc}.

We use auxiliary processes ($\pid r_0, \pid r_1,\ldots$) for intermediate computation, and annotate the encoding
with the index $\ell$ of the first free auxiliary process name ($\enc{f}{\pids p\mapsto\pid q}\ell$).
To alleviate the notation, the encoding assigns mnemonic names to these processes and their correspondence to the
actual process names is formalised in the text using $\pi(f)$ for the number of auxiliary processes needed for
encoding $f:\NN^n\to\NN$, defined by
\begin{align*}
  \pi(S)=\pi(Z)=\pi\left(P^n_m\right)&=0 &
  \pi(R(f,g)) &= \pi(f)+\pi(g)+3 \\
  \pi\left(C(f,g_1,\ldots,g_k)\right) &= \textstyle \pi(f)+\sum_{i=1}^k\pi(g_i)+k &
  \pi(M(f)) &= \pi(f)+3
\end{align*}

For simplicity, we write $\pids p$ for $\pid p_1,\ldots,\pid p_n$ (when $n$ is known) and $\{A_i\}_{i=1}^n$ for
$A_1\gseq\ldots\gseq A_n$.

The encoding of the base cases is straightforward.
\[\enc Z{\pid p\mapsto\pid q}\ell = \com{\pid p.\varepsilon}{\pid q}
  \qquad
  \enc S{\pid p\mapsto\pid q}\ell = \com{\pid p.(\s\cdot\pcont)}{\pid q}
  \qquad
  \enc{P^n_m}{\pids p\mapsto\pid q}\ell = \com{\pid p_m.\pcont}{\pid q}
\]
Composition is also simple.
Let $h=C(f,g_1,\ldots,g_k):\NN^n\to\NN$.
Then:
\begin{align*}
  \enc h{\pids p\mapsto\pid q}\ell =
    &\left\{\enc{g_i}{\pids p\mapsto\pid r'_i}{\ell_i}\right\}_{i=1}^k\gseq\,\enc f{\pid r'_1,\ldots,\pid r'_k\mapsto\pid q}{\ell_{k+1}}
\end{align*}
where $\pid r'_i=\pid r_{\ell+i-1}$, $\ell_1 = \ell+k$ and $\ell_{i+1} = \ell_i+\pi(g_i)$.
Each auxiliary process $\pid r'_i$ connects the output of $g_i$ to the corresponding input of $f$.
Choreographies obtained inductively use these process names as parameters; name clashes are prevented by increasing
$\ell$.
By definition of $\gseq$ $\enc{g_{i+1}}{}{}$ is substituted for the (unique) exit point of
$\enc{g_i}{}{}$, and $\enc f{}{}$ is substituted for the exit point of $\enc{g_k}{}{}$.
The resulting choreography also has only one exit point (that of $\enc f{}{}$).
Below we discuss how to modify this construction slightly so that the $g_i$s are computed in parallel.

For the recursion operator, we need to use recursive procedures.
Let $h=R(f,g):\NN^{n+1}\to \NN$.
Then, using the macro $\macro{inc}$ from Example~\ref{ex:inc_add} for brevity:
\begin{align*}
  \enc h{\pid p_0,\ldots,\pid p_n\mapsto\pid q}\ell ={}\quad
    &\m{def}\,{T} =\m{if}\,(\eqcom{\pid r_c}{\pid p_0})\,\m{then}\,(\com{\pid q'.\pcont}{\pid q};\,\nil)\\
    &\qquad\qquad\m{else}\, \enc g{\pid r_c,\pid q',\pid p_1,\ldots,\pid p_n\mapsto\pid r_t}{\ell_g}\gseq\,
     \com{\pid r_t.\pcont}{\pid q'};\,
     \macro{inc}(\pid r_c,\pid r_t)\gseq\, T\\
    & \m{in}\, \enc f{\pid p_1,\ldots,\pid p_n\mapsto\pid q'}{\ell_f}\gseq\,\com{\pid r_t.\varepsilon}{\pid r_c};\,T
\end{align*}
where
$\pid q' = \pid r_{\ell}$,
$\pid r_c = \pid r_{\ell+1}$,
$\pid r_t = \pid r_{\ell+2}$,
$\ell_f = \ell+3$ and
$\ell_g = \ell_f+\pi(f)$.
Process $\pid r_c$ is a counter, $\pid q'$ stores intermediate results, and $\pid r_t$ is temporary storage; $T$
checks the value of $\pid r_c$ and either outputs the result or recurs.
Note that $\enc h{}{}$ has only one exit point (after the communication from $\pid r$ to $\pid q$), as the exit
points of $\enc f{}{}$ and $\enc g{}{}$ are replaced by code ending with calls to $T$.

The strategy for minimization is similar, but simpler.
Let $h=M(f):\NN^n\to\NN$.
Again we use a counter $\pid r_c$ and compute successive values of $f$, stored in $\pid q'$, until a zero is
found.
This procedure may loop forever, either because $f(\til x,x_{n+1})$ is never $0$ or because one of the evaluations
itself never terminates.
\begin{align*}
  \enc h{\pid p_1,\ldots,\pid p_{n+1}\mapsto\pid q}\ell ={}
    &\m{def}\,{T} =\enc f{\pid p_1,\ldots,\pid p_n,\pid r_c\mapsto\pid q'}{\ell_f}\gseq\,
     \com{\pid r_c.\varepsilon}{\pid r_z}; \\
    &\qquad\m{if}\,(\eqcom{\pid r_z}{\pid q'})\,\m{then}\,(\com{\pid r_c.\pcont}{\pid q};\,\nil)\,
     \m{else}\,(\macro{inc}(\pid r_c,\pid r_z)\gseq\,T)\\
    & \m{in}\,\com{\pid r_z.\varepsilon}{\pid r_c};\,T
\end{align*}
where
$\pid q' = \pid r_{\ell}$,
$\pid r_c = \pid r_{\ell+1}$,
$\pid r_z = \pid r_{\ell+2}$,
$\ell_f = \ell+3$ and
$\ell_g = \ell_f+\pi(f)$.
In this case, the whole if-then-else is inserted at the exit point of
$\enc f{}{}$; the only exit point of this choreography is again after communicating
the result to $\pid q$.

\begin{definition}[Encoding]
  Let $f\in\mathcal R$.
  The \emph{encoding} of $f$ in MC
  is $\enc f{\pids p\mapsto\pid q}{}=\enc f{\pids p\mapsto\pid q}0$.
\end{definition}

\begin{example}
We illustrate this construction by showing the encoding of the $\Sum$
and $\Sub$ functions given in Example~\ref{ex:sumsub}.
Recall that $\Sum=R(P^1_1,C(S,P^3_2))$.
Expanding $\enc{\Sum}{\pid p_x,\pid p_y\mapsto\pid q}{}$
we obtain:
\begin{align*}
  \lefteqn{\enc{\Sum}{\pid p_x,\pid p_y\mapsto\pid q}0 ={}} \\
    &\m{def}\,{T} = \m{if}\,(\eqcom{\pid r_1}{\pid p_x})\,\m{then}\,(\com{\pid r_0.\pcont}{\pid q};\,\nil)\\
  &\qquad\qquad\m{else}\,%
    \underbrace{\com{\pid r_0.\pcont}{\pid r_3}}_{\enc{P^3_2}{\pid r_1,\pid r_0,\pid p_y\mapsto\pid r_3}{4}};\,
    \underbrace{\com{\pid r_3.(\s\cdot\pcont)}{\pid r_2}}_{\enc{S}{\pid r_3\mapsto\pid r_2}{4}};\,
    \com{\pid r_2.\pcont}{\pid r_0};
      \underbrace{\com{\pid r_1.\pcont}{\pid r_2};\,\com{\pid r_2.(\s\cdot\pcont)}{\pid r_1}}_{\macro{inc}(\pid r_1,\pid r_2)};\,
    T\\
    & \m{in}\,%
    \underbrace{\com{\pid p_y.\pcont}{\pid r_0}}_{\enc{P^1_1}{\pid p_y\mapsto\pid r_0}{3}};\,%
    \com{\pid r_2.\varepsilon}{\pid r_1};\,T
\end{align*}
The first two actions in the \m{else} branch correspond to
$\enc{C(S,P^3_2)}{\pid r_1,\pid r_0,\pid p_y\mapsto\pid r_2}3$.

For subtraction, we first show how to implement equality directly in MC, without resorting to its proof of
membership in $\mathcal R$.
This choreography is not the simplest possible because we want it to have only one exit point; its
construction illustrates how any choreography can be transformed to have this property.
\begin{multline*}
\macro{eq}(\pid p_x,\pid p_y,\pid q,\pid r) \defeq{}
  \m{def}\,T=(\com{\pid r.\pcont}{\pid q};\,\nil)\,\m{in}\\
  \m{if}\,(\eqcom{\pid p_x}{\pid p_y})\,\m{then}\,(\com{\pid p_x.\varepsilon}{\pid r};\,T)\,\m{else}\,(\com{\pid p_x.(\suc\pcont)}{\pid r};\,T)
\end{multline*}
Recall now that $\Sub=M(C(\Eq,C(\Sum,P^3_2,P^3_3),P^3_1))$.
Unfolding the encoding of minimization and composition, we obtain
that $\enc{\Sub}{\pid p_x,\pid p_y\mapsto\pid q}0$ is
\begin{align*}
\m{def}\,T ={}
  &\enc{P^3_2}{\pid p_x,\pid p_y,\pid r_1\mapsto\pid r_5}7\gseq\,%
  \enc{P^3_2}{\pid p_x,\pid p_y,\pid r_1\mapsto\pid r_6}7\gseq\,%
  \enc{\Sum}{\pid r_5,\pid r_6\mapsto \pid r_3}7\gseq\\
  &\enc{P^3_1}{\pid p_x,\pid p_y,\pid r_1\mapsto\pid r_4}{11}\gseq\,
  \macro{eq}(\pid r_3,\pid r_4,\pid r_0,\pid r_{11})\gseq\,\com{\pid r_1.\varepsilon}{\pid r_2};\\
  &\m{if}\,\eqcom{\pid r_2}{\pid r_0}\,\m{then}\,(\com{\pid r_1.\pcont}{\pid q};\nil)\,
  \m{else}\,(\macro{inc}(\pid r_1,\pid r_2)\gseq\, T)\\
  \m{in}\,& \com{\pid r_2.\varepsilon}{\pid r_1};\,T
\end{align*}
The first line in the definition of $T$ is
$\enc{C(\Sum,P^3_2,P^3_3)}{\pid p_x,\pid p_y,\pid r_1\mapsto \pid r_3}5$;
the first five processes composed therewithin are
$$\enc{C(\Eq,C(\Sum,P^3_2,P^3_3),P^3_1)}{\pid p_x,\pid p_y,\pid r_1\mapsto \pid r_0}3\,.$$

Fully unfolding the base cases, we obtain
\begin{align*}
\enc{\Sub}{\pid p_x,\pid p_y\mapsto\pid q}0=
  \m{def}\,&T = \com{\pid p_y.\pcont}{\pid r_5};\,\com{\pid r_1.\pcont}{\pid r_6};\\
  &\m{def}\,{R} = \m{if}\,(\eqcom{\pid r_8}{\pid r_5}) \\
  &\quad\m{then}\,\com{\pid r_7.\pcont}{\pid r_3};\,\com{\pid p_x.\pcont}{\pid r_4};\\
  &\qquad\m{def}\,S=\com{\pid r_{11}.\pcont}{\pid r_0};\,\com{\pid r_1.\varepsilon}{\pid r_2};\\
  &\qquad\quad\m{if}\,(\eqcom{\pid r_2}{\pid r_0})\,\m{then}\,(\com{\pid r_1.\pcont}{\pid q};\,\nil)\,
   \m{else}\,(\com{\pid r_1.\pcont}{\pid r_2};\,\com{\pid r_2.(\s\cdot\pcont)}{\pid r_1};\,T)\\
  &\qquad\m{in}\,\m{if}\,(\eqcom{\pid r_3}{\pid r_4})\,\m{then}\,(\com{\pid r_3.\varepsilon}{\pid r_{11}};\,S)\,
   \m{else}\,(\com{\pid r_3.(\suc\pcont)}{\pid r_{11}};\,S)\\
  &\quad\m{else}\,\com{\pid r_7.\pcont}{\pid r_{10}};\,\com{\pid r_{10}.(\s\cdot\pcont)}{\pid r_9};\\
  &\qquad\com{\pid r_9.\pcont}{\pid r_7};\,\com{\pid r_8.\pcont}{\pid r_9};\,\com{\pid r_9.(\s\cdot\pcont)}{\pid r_8};\, R\\
  &\m{in}\,\com{\pid r_6.\pcont}{\pid r_7};\,\com{\pid r_9.\varepsilon}{\pid r_8};\,R\\
  \m{in}\,&\com{\pid r_2.\varepsilon}{\pid r_1};\,T
\end{align*}
Due to the way sequential composition works, the structure of the
definition of $\Sub$ is no longer clear in this fully unfolded
encoding.
\end{example}

\subsection{Soundness and Main Results}
\label{sec:sound}

By induction we now show that the construction presented above is sound.
In the proof, we use partial specifications of states.
For example, $C,\sstate{\pid p\mapsto v}\to C',\sstate{\pid q\mapsto w}$ denotes that execution of $C$ from
\emph{any} state where $\pid p$ contains value $v$ will yield $C'$ in \emph{some} state where $\pid q$
contains value $w$.

\begin{theorem}[Turing completeness of MC]
  \label{teo:sound1}
  If $f:\NN^n\to\NN$ and $f\in\mathcal R$, then, for every
  $k$, $\enc f{\pids p\mapsto\pid q}k$ implements $f$ with input
  processes $\pids p=\pid p_1,\dots,\pid p_n$ and output process $\pid q$.
\end{theorem}
\begin{proof}
  The proof is by induction on the definition of the set of partial
  recursive functions.
  We use a stronger induction hypothesis -- namely, that if
  $\sigma(\pid p_i)=\numeral{x_i}$ and $f(\til x)$ is defined,
  then $\enc f{\pids p\mapsto\pid q}k,\sigma\to^\ast\sigma'$ where
  $\sigma'(\pid p_i)=\numeral{x_i}$ and
  $\sigma'(\pid q)=\numeral{f(\til x)}$.
  The extra assumption that the input values are not changed during
  execution is essential for the inductive step.
  In the case where $f(\til x)$ is not defined, we assume as before that
  $\enc f{\pids p\mapsto\pid q}k,\sigma\not\to^\ast\nil$.
  \begin{enumerate}
  \item For each base case, it is straightforward to compute the
    sequence of reductions from the rules and the definition of the
    corresponding choreography.
    We exemplify this with successor.
  \[
    \enc S{\pid p\mapsto\pid q}\ell:
    \com{\pid p.(\s\cdot\pcont)}{\pid q},\sstate{\pid p\mapsto\numeral{x}}
    \to\nil,\sstate{\pid p\mapsto\numeral{x}\\ \pid q\mapsto\numeral{x+1}}
    \]
\item Let $h=C(f,g_1,\ldots,g_k):\NN^n\to\NN$.
  The result follows directly from the induction hypothesis and Lemma~\ref{lem:seqcomp}.
\item Let $h=R(f,g):\NN^{n+1}\to\NN$.
  By induction hypothesis, choreographies
  $\enc f{\pid p_1,\ldots,\pid p_n\mapsto\pid q}{\ell_f}$ and
  $\enc{g}{\pid p_1,\ldots,\pid p_{n+2}\mapsto\pid q}{\ell_g}$
  implement $f$ and $g$, respectively, for all $\pids p$, $\pid q$, $\ell_f$ and
  $\ell_g$.
  Again, assume first that $h(x_0,\til x)$ is defined.
  Then:
  \begin{align*}
    \enc h{\pid p_0,\pids p\mapsto\pid q}\ell:\quad
    &\m{def}\,{T}=(\ldots)\,\m{in}\, \enc f{\pids p\mapsto\pid q'}{\ell_f}\gseq\com{\pid r_t.\varepsilon}{\pid r_c};\,T,\sstate{\pid p_i\mapsto\numeral{x_i}}\\
    \lto{IH}^\ast{}&
    \m{def}\,{T}=(\ldots)\,\m{in}\, \com{\pid r_t.\varepsilon}{\pid r_c};\,T,%
    \sstate{\pid p_i\mapsto\numeral{x_i}\\ \pid q'\mapsto\numeral{f(\til x)}}\\
    \to{}&\m{def}\,{T}=(\ldots)\,\m{in}\, T,%
    \sstate{\pid p_i\mapsto\numeral{x_i}\\ \pid q'\mapsto\numeral{h(0,\til x)}\\ \pid r_c\mapsto\numeral{0}}
  \end{align*}

  We now prove that
  \[
    \m{def}\,{T}=(\ldots)\,\m{in}\, T,%
    \sstate{\pid p_i\mapsto\numeral{x_i}\\ \pid q'\mapsto\numeral{h(k,\til x)}\\ \pid r_c\mapsto\numeral{k}}
    \to^\ast
    \m{def}\,{T}=(\ldots)\,\m{in}\, T,%
    \sstate{\pid p_i\mapsto\numeral{x_i}\\ \pid q'\mapsto\numeral{h(k+1,\til x)}\\ \pid r_c\mapsto\numeral{k+1}}
  \]
  for all $k<x_0$.
  We only need to unfold $T$ once, so we omit the
  $\m{def}\,T=(\ldots)\,\m{in}$ wrapper in the next reduction sequence.

  Since $k<x_0$, the definition of $T$ reduces to the \m{else} branch:
  \begin{align*}
    T \to^\ast{}&\enc g{\pid r_c,\pid q',\pids p\mapsto\pid r_t}{\ell_g}\gseq\,\com{\pid r_t.\pcont}{\pid q'};\,
      \com{\pid r_c.\pcont}{\pid r_t};\,\com{\pid r_t.(\s\cdot\pcont)}{\pid r_c};\, T,
      \sstate{\pid p_i\mapsto\numeral{x_i}\\ \pid q'\mapsto\numeral{h(k,\til x)}\\ \pid r_c\mapsto\numeral k}\\
      \lto{IH}^\ast{}&\com{\pid r_t.\pcont}{\pid q'};\,%
      \com{\pid r_c.\pcont}{\pid r_t};\,\com{\pid r_t.(\s\cdot\pcont)}{\pid r_c};\, T,
      \sstate{\pid p_i\mapsto\numeral{x_i}\\ \pid q'\mapsto\numeral{h(k,\til x)}\\ \pid r_c\mapsto\numeral k\\ \pid r_t\mapsto\numeral{g(k,h(k,\til x),\til x)}}\\
      \to^\ast{} &T,\sstate{\pid p_i\mapsto\numeral{x_i}\\ \pid q'\mapsto\numeral{h(k+1,\til x)}\\ \pid r_c\mapsto\numeral{k+1}\\ \pid r_t\mapsto\numeral k}
  \end{align*}
  which establishes the thesis, ignoring the value in $\pid r_t$.

  By induction on $x_0$ we obtain that
  \begin{align*}
    \enc h{\pid p_0,\pids p\mapsto\pid q}\ell,\sstate{\pid p_i\mapsto\numeral{x_i}}
    \to^\ast{}&\m{def}\,T=(\ldots)\,\m{in}\,%
    T,\sstate{\pid p_i\mapsto\numeral{x_i}\\ \pid q'\mapsto\numeral{h(x_0,\til x)}\\ \pid r_c\mapsto\numeral{x_0}}\\
    \lto{(1)}{}&\m{def}\,T=(\ldots)\,\m{in}\,%
    \com{\pid q'.\pcont}{\pid q};\,\nil%
    ,\sstate{\pid p_i\mapsto\numeral{x_i}\\ \pid q'\mapsto\numeral{h(x_0,\til x)}\\ \pid r_c\mapsto\numeral{x_0}}\\
    \to^\ast{}&\m{def}\,T=(\ldots)\,\m{in}\,\nil%
    ,\sstate{\pid p_i\mapsto\numeral{x_i}\\ \pid q'\mapsto\numeral{h(x_0,\til x)}\\ \pid r_c\mapsto\numeral{x_0}\\ \pid q\mapsto\numeral{h(x_0,\til x)}}
  \end{align*}
  and the last process is equivalent to $\nil$.
  In $(1)$ we used the fact that the contents of $\pid r_c$ and $\pid p_0$ are
  both equal to $\numeral{x_0}$.

  If $h(x_0,\til x)$ is not defined, there are two possible cases.
  If $f(\til x)$ is not defined, then
  $\enc f{\pids p\mapsto\pid q'}{\ell_f}$ diverges from any
  state where each $\pid p_i$ contains $\numeral{x_i}$, whence so does
  $\enc h{p_0,\pids p\mapsto\pid q}{\ell}$ by
  Lemma~\ref{lem:seqcomp} and rule~\rname{C}{Ctx}.
  If $g(k,h(k,\til x),\til x)$ is undefined for some $k<x_0$, then
  divergence is likewise obtained from the fact that
  $\enc g{\pid r_c,\pid q',\pids p}{\ell_g}$ diverges from any state
  where $\pid r_c$ contains \numeral{k}, $\pid q'$ contains
  \numeral{h(k,\til x)}, and $\pid p_i$ contains $\numeral{x_i}$.

  \item The case where $h=M(f):\NN^n\to\NN$ is very similar, the
    auxiliary result now stating that
  \[
    \m{def}\,{T}=(\ldots)\,\m{in}\, T,%
    \sstate{\pid p_i\mapsto\numeral{x_i}\\ \pid r_c\mapsto\numeral k}
    \to^\ast
    \m{def}\,{T}=(\ldots)\,\m{in}\, T,%
    \sstate{\pid p_i\mapsto\numeral{x_i}\\ \pid r_c\mapsto\numeral{k+1}}
  \]
  as long as $f(\til x,k)$ is defined and different from $0$.

  The only new aspect is that non-termination may arise from the
  fact that $f(\til x,k)$ is defined and non-zero for every $k\in\NN$,
  in which case we get an infinite reduction sequence
  \begin{align*}
    \enc h{\pids p\to\pid q}\ell,\sstate{\pid p_i\mapsto\numeral{x_i}}
    \to^\ast{}&\m{def}\,{T}=(\ldots)\,\m{in}\, T,%
    \sstate{\pid p_i\mapsto\numeral{x_i}\\ \pid r_c\mapsto\numeral0}\\
    \to^\ast{}&
    \m{def}\,{T}=(\ldots)\,\m{in}\, T,%
    \sstate{\pid p_i\mapsto\numeral{x_i}\\ \pid r_c\mapsto\numeral n}\\
    \to^\ast{}&\ldots
  \end{align*}
  \end{enumerate}
\end{proof}

Since MC is a fragment of CC, this result trivially implies Turing completeness of CC.
Let $\text{\CPP{}}=\{\epp{C,\sigma}{}\mid\epp{C,\sigma}{}\mbox{ is defined}\}$
be the set of the projections of all projectable choreographies in CC.
By Theorem~\ref{cor:df-sp}, all terms in \CPP{} are deadlock-free.
By Lemma~\ref{lem:amend}, for every function $f$ we can amend $\enc f{}{}$
to an equivalent projectable choreography.
Then \CPP{} is Turing complete by Theorems~\ref{thm:oc-cc-sp} and~\ref{teo:sound1}.

\begin{corollary}[Turing completeness of \CPP{}]
  \label{cor:cp-tc}
  Every partial recursive function is implementable in \CPP{}.
\end{corollary}
\begin{proof}
Let $f\in\RR$. By Theorem~\ref{teo:sound1}, $C=\enc f{\pids p \mapsto \pid q}{}$ for any suitable
$\pids p$ and $\pid q$
implements $f$. By Lemma~\ref{lem:amend}, $\amend(C)$ is projectable and operationally equivalent to $C$.
Hence, by Theorem~\ref{thm:oc-cc-sp}, $\epp{\amend(C),\sigma}{}$ is a term in SP that correctly implements $f$.
\end{proof}

We finish this section by showing how to optimize our encoding and obtain parallel process
implementations of independent computations. If $h$ is defined by composition
from $f$ and $g_1,\ldots,g_k$, then in principle the computation of the $g_i$s
could be completely parallelised. However, $\enc{}{}{}$ does not fully
achieve this, as $\enc{g_1}{}{}$,\ldots,$\enc{g_k}{}{}$ share the processes containing the input.
We define a modified variant $\altenc{}{}{}$ of $\enc{}{}{}$ where, for
$h=C(f,g_1,\ldots,g_k)$, $\altenc h{\pids p\mapsto\pid q}\ell$ is
\[
\left\{\com{\pid p_j.\pcont}{\pid p_j^i}\right\}_{1\leq i\leq k,1\leq j\leq n}\gseq
\left\{\altenc{g_i}{\pids{p^i}\mapsto\pid r'_i}{\ell_i}\right\}_{i=1}^k\gseq\,
\altenc f{\pid r'_1,\ldots,\pid r'_k\mapsto\pid q}{\ell_{k+1}}
\]
with a suitably adapted label function $\ell$.
Now Theorem~\ref{thm:par} applies, yielding:

\begin{theorem}
  \label{thm:par-compute}
  Let $h=C(f,g_1,\ldots,g_k)$.
  For all $\pids p$ and $\pid q$, if $h(\til x)$ is defined and $\sigma$ is
  such that $\sigma(\pid p_i)=\numeral{x_i}$, then all the
  $\altenc{g_i}{\pids{p^i}\mapsto\pid r'_i}{\ell_i}$
  run in parallel in $\altenc h{\pids p\mapsto\pid q}{}$.
\end{theorem}
\begin{proof}
  By induction on $k$.
  For $k=1$ there is nothing to prove.
  By definition of $\altenc{}{}{}$,
  $\altenc{g_{k+1}}{\pids{p^{k+1}}\mapsto\pid r'_{k+1}}{\ell_{k+1}}$ and
  $\altenc{g_i}{\pids{p^i}\mapsto\pid r'_i}{\ell_i}$ do not share any process names for $i\leq k$,
  hence Theorem~\ref{thm:par} implies that 
  $\altenc{g_{k+1}}{\pids{p^{k+1}}\mapsto\pid r'_{k+1}}{\ell_{k+1}}$ and
  $\left\{\altenc{g_i}{\pids{p^i}\mapsto\pid r'_i}{\ell_i}\right\}_{i=1}^k$ run in parallel in
  $\left\{\altenc{g_i}{\pids{p^i}\mapsto\pid r'_i}{\ell_i}\right\}_{i=1}^{k+1}$,
  and hence in $\altenc h{\pids p\mapsto\pid q}{}$.
\end{proof}

This parallelism is preserved by EPP, through Theorem~\ref{thm:oc-cc-sp}.


%% file: selection.tex

As we saw earlier, MC is the fragment of CC that does not contain selections, and MC choreographies can be
amended into projectable CC choreographies.
We now show that selections are not necessary to ensure projectability: we can encode
those selections introduced by amendment using only conditionals and extra communications.

\subsection{Selections as Value Communications}
Selections are used pervasively for projectability in previous choreography languages, so the fact 
that they are technically unnecessary is both interesting and somewhat unexpected.
This construction increases the size of the choreography exponentially, but both the number of 
processes and
the size of the endpoint projections grow only by a linear factor.

We motivate our construction with an example.

\begin{example}
Consider a choreography where $\pid p$ makes a choice depending on the value stored by $\pid q$, and then
$\pid r$ needs to be notified of the result (because, e.g., it is involved in further communications in one or both
of the branches).
As an example, we take $C$ to be the choreography
$\cond{\eqcom{\pid p}{\pid q}}{\sel{\pid p}{\pid r}{\lleft};C_1}{\sel{\pid p}{\pid r}{\lright};C_2}$, where $\pid r$ 
has different behaviours in $C_1$ and $C_2$.

In order to eliminate the label selections, $\pid r$ must be able to perform a conditional that is guaranteed to
choose the same branch as taken by $\pid p$.
With this in mind, we introduce an auxiliary process $\pid p^\ast$ and add communications from $\pid p$ to $\pid p^\ast$
of $\zero$ (\m{then} branch) or $\s\cdot\pcont$ (\m{else} branch).
Then $\pid r$ can recover this information by first setting its contents to $\zero$ and then comparing them with
$\pid p^\ast$; this requires another auxiliary process $\pid r^\ast$ to store $\pid r$'s value in the meantime.
Furthermore, even though we know at a global level what the result of the comparison will be, the EPP
(in particular, merging) demands that we consider both branches in both cases.
We therefore rewrite $C$ as follows.
\begin{align*}
\hspace*{-1em}\m{if}\,{\eqcom{\pid p}{\pid q}}\,
&\m{then}\left(\,%
  \com{\pid p.\zero}{\pid p^\ast};\,
  \com{\pid r.\pcont}{\pid r^\ast};\,
  \com{\pid r^\ast.\zero}{\pid r};\,
  \cond{\eqcom{\pid r}{\pid p^\ast}}{%
    \left(\com{\pid r^\ast.\pcont}{\pid r};\,C_1\right)%
  }{%
    \left(\com{\pid r^\ast.\pcont}{\pid r};\,C_2\right)%
}\right)\\
&\m{else}\left(\,%
  \com{\pid p.\s\cdot\pcont}{\pid p^\ast};\,
  \com{\pid r.\pcont}{\pid r^\ast};\,
  \com{\pid r^\ast.\zero}{\pid r};\,
  \cond{\eqcom{\pid r}{\pid p^\ast}}{%
    \left(\com{\pid r^\ast.\pcont}{\pid r};\,C_1\right)%
  }{%
    \left(\com{\pid r^\ast.\pcont}{\pid r};\,C_2\right)%
\right)}
\end{align*}
Observe that the behaviour of the processes not performing conditionals ($\pid p^\ast$ and $\pid r^\ast$) is the same
in all four branches,
while $\pid p$ and $\pid r$ have two possible behaviours that are independent of each other's choices.
This guarantees that merging will work for all projections.
\qed
\end{example}

Recall that the definition of amendment guarantees that selections only occur in branches of conditionals, and
that they are always paired and in the same order.
These properties are essential to our construction.
The fragment of CC obtained by amending choreographies in MC can be inductively generated by
\[
  C ::= \gencom; C \acspar \gencondmin \acspar \genrec \acspar \gencall \acspar \nil
\]
where $S(\pid p,\pids r,\ell,C)$ prepends selections of label $\ell$ from $\pid p$ to all processes in the
list $\pids r$.
Formally, $S$ is defined as
\[
  S(\pid p,\emptyset,\ell,C) = C
  \qquad
  S(\pid p,\pid r::\pids r,\ell,C) = \sel{\pid p}{\pid r}\ell;\,S(\pid p,\pids r,\ell,C)
\]

\begin{definition}[Selection elimination]
  Let $C$ be a choreography obtained by amending a choreography in MC.
  The encoding $\delsel C^+$
  of $C$ in MC uses processes $\pid p,\pidb p$ for each $\pid p\in\pn(C)$, 
  plus a special process $\pid z$, and is defined in Figure~\ref{fig:delsel}.
  \input{delsel}
\end{definition}
\newtext
The auxiliary function $\delsel{C_1,C_2}$ is used to eliminate selections in conditionals, by simultaneously traversing both branches.
As usual, we write $\delsel{C_1,C_2}_1$ to denote the first component of the resulting pair, and likewise for the second component.
\oldtext

This definition significantly exploits the structure of amended choreographies, where selections are always
paired at the top of the two branches of conditionals.
It follows from it that $|\pn(\delsel{C}^+)|=2|\pn(C)|+1$ and that $|\delsel{C}^+|\leq 2^{|C|}$.
However, the EPP from MC to MP collapses all branches of conditionals, hence
$|\epp{\delsel{C}^+}{\pidb q}|\leq|\epp{\delsel{C}^+}{\pid q}|\leq 3|\epp{C}{\pid q}|$ for every
$\pid q\in\pn(C)$.

\begin{theorem}[Selection elimination]
  \label{thm:delsel-wd}
  For every choreography $C\in\mbox{MC}$, $\epp{\delsel{\amend(C)}}{}$ is defined.
\end{theorem}

For convenience, we split the proof of this result in several lemmas.

\begin{lemma}
  \label{lem:delsel1}
  If $\pid q\in\pids r$, then
  $\epp{\delselp{S(\pid p,\pids r,\lleft,C_1)}{S(\pid p,\pids r,\lright,C_2)}_1}{\pid q}=
  \epp{\delselp{S(\pid p,\pids r,\lleft,C_1)}{S(\pid p,\pids r,\lright,C_2)}_2}{\pid q}$.
\end{lemma}
\begin{proof}
  By induction on the length of $\pids r$.
  If $\pids r=\emptyset$, then the result is vacuously true.
  If $\pids r$ does not start with $\pid q$, then the result follows trivially from the induction hypothesis.
  So consider the case where $\pids r=\pid q::\pids{r'}$.
  In this case,
  $\delselp{S(\pid p,\pids r,\lleft,C_1)}{S(\pid p,\pids r,\lright,C_2)}$ unfolds to
  \begin{align*}
    \langle &
      \com{\pid q.\pcont}{\pidb q};\,\com{\pid p.\varepsilon}{\pid q};\,
      \m{if}\,\eqcom{\pid q}{\pid z}\,
      \m{then}\,\com{\pidb q.\pcont}{\pid q};\,\delselp{S(\pid p,\pids{r'},\lleft,C_1)}{S(\pid p,\pids{r'},\lright,C_2)}_1\\
      &\hspace*{37mm}
      \m{else}\,\com{\pidb q.\pcont}{\pid q};\,\delselp{S(\pid p,\pids{r'},\lleft,C_1)}{S(\pid p,\pids{r'},\lright,C_2)}_2,\\
      & \com{\pid q.\pcont}{\pidb q};\,\com{\pid p.\s\pcont}{\pid q};\,
      \m{if}\,\eqcom{\pid q}{\pid z}\,
      \m{then}\,\com{\pidb q.\pcont}{\pid q};\,\delselp{S(\pid p,\pids{r'},\lleft,C_1)}{S(\pid p,\pids{r'},\lright,C_2)}_1\\
      &\hspace*{39mm}
      \m{else}\,\com{\pidb q.\pcont}{\pid q};\,\delselp{S(\pid p,\pids{r'},\lleft,C_1)}{S(\pid p,\pids{r'},\lright,C_2)}_2\rangle
  \end{align*}
  and the endpoint projections of both choreographies for $\pid q$ become
    \begin{multline*}
      \asend{\pidb q}{\pcont};\,\arecv{\pid p};\,\m{if}\,\eqcom{\pcont}{\pid z}\,
    \m{then}\,\arecv{\pidb q};\,\epp{\delselp{S(\pid p,\pids{r'},\lleft,C_1)}{S(\pid p,\pids{r'},\lright,C_2)}_1}{\pid q}\\
    \m{else}\,\arecv{\pidb q};\,\epp{\delselp{S(\pid p,\pids{r'},\lleft,C_1)}{S(\pid p,\pids{r'},\lright,C_2)}_2}{\pid q}
    \end{multline*}
    which are defined and identical.
\end{proof}

\begin{lemma}
  \label{lem:delsel2}
  If $\epp{\delsel{C_1}}{\pid q}=\epp{\delsel{C_2}}{\pid q}$ and $\pid p\neq\pid q\not\in\pids r$, then
  $\epp{\delselp{S(\pid p,\pids r,\lleft,C_1)}{S(\pid p,\pids r,\lright,C_2)}_1}{\pid q}=
  \epp{\delselp{S(\pid p,\pids r,\lleft,C_1)}{S(\pid p,\pids r,\lright,C_2)}_2}{\pid q}$.
\end{lemma}
\begin{proof}
  By induction on the length of $\pids r$.
  If $\pids r=\emptyset$, then the result reduces to the hypothesis.
  Otherwise,
  $\delselp{S(\pid p,\pid r::\pids r,\lleft,C_1)}{S(\pid p,\pid r::\pids r,\lright,C_2)}$ unfolds to
  \begin{align*}
    \langle &
      \com{\pid r.\pcont}{\pidb r};\,\com{\pid p.\varepsilon}{\pid r};\,
      \m{if}\,\eqcom{\pid r}{\pid z}\,
      \m{then}\,\com{\pidb r.\pcont}{\pid r};\,\delselp{S(\pid p,\pids r,\lleft,C_1)}{S(\pid p,\pids r,\lright,C_2)}_1\\
      &\hspace*{37mm}
      \m{else}\,\com{\pidb r.\pcont}{\pid r};\,\delselp{S(\pid p,\pids r,\lleft,C_1)}{S(\pid p,\pids r,\lright,C_2)}_2,\\
      & \com{\pid r.\pcont}{\pidb r};\,\com{\pid p.\s\pcont}{\pid r};\,
      \m{if}\,\eqcom{\pid r}{\pid z}\,
      \m{then}\,\com{\pidb r.\pcont}{\pid r};\,\delselp{S(\pid p,\pids r,\lleft,C_1)}{S(\pid p,\pids r,\lright,C_2)}_1\\
      &\hspace*{39mm}
      \m{else}\,\com{\pidb r.\pcont}{\pid r};\,\delselp{S(\pid p,\pids r,\lleft,C_1)}{S(\pid p,\pids r,\lright,C_2)}_2\rangle
  \end{align*}
  and since $\pid q\neq\pid r$ there are three cases to consider.
  \begin{itemize}
  \item $\pid q$ is $\pidb r$: then both endpoint projections become
    \[\arecv{\pid r};\,
    \left(\asend{\pid r}\pcont;\,\epp{\delselp{S(\pid p,\pids r,\lleft,C_1)}{S(\pid p,\pids r,\lright,C_2)}_1}{\pidb q}\right)
    \merge\left(\asend{\pid r}\pcont;\,\epp{\delselp{S(\pid p,\pids r,\lleft,C_1)}{S(\pid p,\pids r,\lright,C_2)}_2}{\pidb q}\right)\]
    and by induction hypothesis the two processes being merged are identical, so the result is defined.
  \item $\pid q$ is $\pid z$: then both endpoint projections become
    \[\asend{\pid r}\pcont;\,
    \epp{\delselp{S(\pid p,\pids r,\lleft,C_1)}{S(\pid p,\pids r,\lright,C_2)}_1}{\pid z}
    \merge\epp{\delselp{S(\pid p,\pids r,\lleft,C_1)}{S(\pid p,\pids r,\lright,C_2)}_2}{\pid z}\]
    and again by induction hypothesis the two processes being merged are identical, so the result is defined.
  \item $\pid q$ is another process: then both endpoint projections become simply
    \[\epp{\delselp{S(\pid p,\pids r,\lleft,C_1)}{S(\pid p,\pids r,\lright,C_2)}_1}{\pid q}
    \merge\epp{\delselp{S(\pid p,\pids r,\lleft,C_1)}{S(\pid p,\pids r,\lright,C_2)}_2}{\pid q}\]
    whence the induction hypothesis guarantees again that the two processes being merged are identical, so the result is defined.\qedhere
  \end{itemize}
\end{proof}

\begin{lemma}
  \label{lem:delsel}
  For every choreography $C$ in MC and every process $\pid r$, $\epp{\delsel{\amend(C)}}{\pid r}$ is defined.
\end{lemma}
\begin{proof}
  By structural induction on $\amend(C)$.
  The only non-trivial case is that where $\amend(C)$ is $\gencondmin$, where we need to consider the possible
  cases for $\pid r$.
  If $\pid r=\pid p$, then the induction hypothesis establishes the thesis with induction over $\pids r$.
  If $\pid r\in\pids r$, then Lemma~\ref{lem:delsel1} guarantees that both branches of the conditional will be
  equal, hence the endpoint projection is again defined.
  Finally, if $\pid r\not\in\pids r$, then by definition of amendment
  $\epp{\amend(C_1)}{\pid r}=\epp{\amend(C_2)}{\pid r}$, whence Lemma~\ref{lem:delsel2} applies and
  establishes the thesis as in the previous case.
\end{proof}

\begin{proof}[Proof (Theorem~\ref{thm:delsel-wd}).]
  Straightforward consequence of Lemma~\ref{lem:delsel}.
\end{proof}

The operational semantics of $C$ and $\delsel{\amend(C)}$ are related by the following results, which are
straightforward to prove by structural induction.

\begin{lemma}
  Choreographies $C$ and $\delsel{\amend(C)}^+$ are equivalent wrt $\pn(C)$.
\end{lemma}

\begin{lemma}
  If $C,\sigma\to C',\sigma'$ and $\sigma^+$ is such that $\sigma^+(\pid p)=\sigma(\pid p)$ for
  $\pid p\in\pn(C)$ and $\sigma^+(\pid z)=\varepsilon$, then
  $\delsel{\amend(C)},\sigma^+\to^\ast\delsel{C'},{\sigma'}^+$ for some ${\sigma'}^+$ similarly related to $\sigma'$.
  Furthermore, the latter reduction consists of only one step except for the case when the former uses rule
  \rname{C}{Cond}.

  Conversely, if $\delsel{\amend(C)},\sigma^+\to C',\sigma'$, then $C,\sigma\to C'',\sigma''$ where
  $C',\sigma'\to^\ast\delsel{\amend(C'')},{\sigma''}^+$.
  Furthermore, the latter reduction is non-empty only in the case when the former uses rule \rname{C}{Cond}.
\end{lemma}

\begin{corollary}
  With the notation of the previous lemma, if $C,\sigma\to^\ast C',\sigma'$, then
  $\delsel{\amend(C)}^+,\sigma^+\to^\ast\delsel{\amend(C')},{\sigma'}^+$.
\end{corollary}

As a consequence, the set $\text{\MPP{}}=\{\epp{C,\sigma}{}\mid\epp{C,\sigma}{}\mbox{ is defined}\}$ of
projections of minimal choreographies is also Turing complete.
\begin{corollary}[Turing completeness of \MPP{}]
  \label{cor:mp-tc}
  Every partial recursive function is implementable in \MPP{}.
\end{corollary}

\subsection{Discussion}

\begin{figure}
  \centering
  \resizebox{\textwidth}{!}{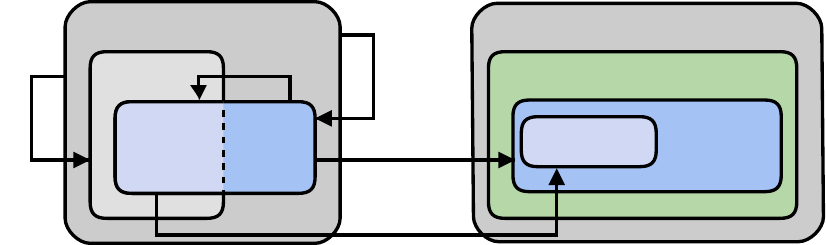}
%
%
  \caption{Summary of the different mappings between calculi.}
  \label{fig:inclusions}
\end{figure}

\newtext
Figure~\ref{fig:inclusions} displays the connections among the different calculi that we studied 
in this work. We now discuss some consequences of these connections.
\oldtext

The results in this section show that label selection is not a necessary primitive in a 
choreography calculus, and thus we could take MC (rather than CC) as our core choreography language.
Furthermore, the construction in \S~\ref{sec:encoding} shows that selections are not needed for implementing
computable functions in CC; they are used only for obtaining projectable choreographies, via amendment.

There is however a strong argument for including label selections in a core choreography calculus.
The advantages of eliminating selections are a simpler choreography language, a simpler definition of EPP
(without merging), and a simpler process language (without selection and branching).
The main drawback is that eliminating a selection needed for projectability makes the choreography  exponentially
larger and requires the addition of extra processes and communications; this significantly changes the structure of
the choreography, potentially making it unreadable.
Selections are also present in virtually all choreography models~\cite{BCDLDY08,CHY12,CM13,HYC08,PGGLM15,QZCY07},
therefore we believe that a core model such as CC should have them (in addition to the drawback we mentioned).

\newtext
The combination of our results on amendment and selection elimination 
suggests the viability of a particular implementation strategy for choreographic programming. 
Programmers could write choreographies without label selections, ignoring how information about 
control flow is propagated between processes. Then, our results could be used to translate 
these choreographies to process implementations in a simple language that does not include label 
communications (like \MPP{}). This would simplify the target language, since it would not require
primitives for selection and branching. The exponential growth of the intermediate choreography 
representation can be bypassed by using shared data structures for the syntax tree, since the 
generated choreographies contain a lot of duplicate terms.
\oldtext

\newtext
However, such a methodology makes use of amendment, and relying exclusively on amendment removes an 
important ability provided in CC and all other standard choreography
calculi: deciding at which point of execution selections should be performed.
In more expressive languages than CC, processes can perform complex internal 
computations~\cite{ourPCstuff}.
For example, assume that $\pid p$ had to assign tasks to other two processes $\pid r$ and $\pid s$ based on a 
condition. In one case, $\pid r$ would run a slow task and $\pid s$ a fast one; otherwise, $\pid r$ would run a fast 
task and $\pid s$ a slow one. In this case, $\pid p$ should begin by sending a selection to the process with the slow 
task and then by sending it the necessary data for its computation, before it sends the selection to the process with 
the fast task.
(Note that it is only amendment that causes this issue, since our selection elimination procedure 
preserves the control flow specified by a choreography with selections.)
\oldtext


%% file: delsel.tex
\begin{figure}
  \begin{align*}
    &\delsel C^+ = \com{\pid p.\varepsilon}{\pid z};\,\delsel C \\[2ex]
    \hline\\
    &\delsel\nil = \nil
    \quad\delsel{\gencom;C} = \gencom;\delsel C
    \quad\delsel{\genrec} = \rec{X}{\delsel{C_2}}{\delsel{C_1}}
    \\
    &\delsel{\gencall} = \gencall
    \qquad\delsel{\gencond} = \m{if}\,\eqcom{\pid p}{\pid q}\,
                              \m{then}\,\delselp{C_1}{C_2}_1\,\m{else}\,\delselp{C_1}{C_2}_2\\[2ex]
    \hline\\
    &\delselp{C_1}{C_2} = \langle\delsel{C_1},\delsel{C_2}\rangle
                          \mbox{ if $C_1$ and $C_2$ do not begin with a selection} \\
    &\delselp{\sel{\pid p}{\pid q}{\lleft};\,C_1}{\sel{\pid p}{\pid q}{\lright};\,C_2} =\\
    &\quad\left\langle
                            \com{\pid q.\pcont}{\pidb q};\,\com{\pid p.\varepsilon}{\pid q};\,
                            \m{if}\,\eqcom{\pid q}{\pid z}\,
                            \m{then}\,\com{\pidb q.\pcont}{\pid q};\,\delselp{C_1}{C_2}_1\,
                            \m{else}\,\com{\pidb q.\pcont}{\pid q};\,\delselp{C_1}{C_2}_2,\right.\\
    &\qquad\left.
                            \com{\pid q.\pcont}{\pidb q};\,\com{\pid p.\s\pcont}{\pid q};\,
                            \m{if}\,\eqcom{\pid q}{\pid z}\,
                            \m{then}\,\com{\pidb q.\pcont}{\pid q};\,\delselp{C_1}{C_2}_1\,
                            \m{else}\,\com{\pidb q.\pcont}{\pid q};\,\delselp{C_1}{C_2}_2\right\rangle
  \end{align*}
\caption{Elimination of selections from amended choreographies.}
\label{fig:delsel}
\end{figure}

%% file: inclusions.pdf_tex
\begingroup%
  \makeatletter%
  \providecommand\color[2][]{%
    \errmessage{(Inkscape) Color is used for the text in Inkscape, but the package 'color.sty' is not loaded}%
    \renewcommand\color[2][]{}%
  }%
  \providecommand\transparent[1]{%
    \errmessage{(Inkscape) Transparency is used (non-zero) for the text in Inkscape, but the package 'transparent.sty' is not loaded}%
    \renewcommand\transparent[1]{}%
  }%
  \providecommand\rotatebox[2]{#2}%
  \ifx\svgwidth\undefined%
    \setlength{\unitlength}{396.044856bp}%
    \ifx\svgscale\undefined%
      \relax%
    \else%
      \setlength{\unitlength}{\unitlength * \real{\svgscale}}%
    \fi%
  \else%
    \setlength{\unitlength}{\svgwidth}%
  \fi%
  \global\let\svgwidth\undefined%
  \global\let\svgscale\undefined%
  \makeatother%
  \begin{picture}(1,0.29693606)%
    \put(0,0){\includegraphics[width=\unitlength]{inclusions.pdf}}%
    \put(0.25062069,0.25450084){\color[rgb]{0,0,0}\makebox(0,0)[b]{\smash{Core Choreographies (CC)}}}%
    \put(0.77581378,0.25450084){\color[rgb]{0,0,0}\makebox(0,0)[b]{\smash{Stateful Processes (SP)}}}%
    \put(0.77379387,0.19390169){\color[rgb]{0,0,0}\makebox(0,0)[b]{\smash{Deadlock-free SP}}}%
    \put(0.71521458,0.11310272){\color[rgb]{0,0,0}\makebox(0,0)[b]{\smash{$\MPP{}$}}}%
    \put(0.86671257,0.11310272){\color[rgb]{0,0,0}\makebox(0,0)[b]{\smash{$\CPP{}$}}}%
    \put(0.20012136,0.19390165){\color[rgb]{0,0,0}\makebox(0,0)[b]{\smash{MC}}}%
    \put(0.26072056,0.13330245){\color[rgb]{0,0,0}\makebox(0,0)[b]{\smash{Projectable}}}%
    \put(0.26072056,0.09290299){\color[rgb]{0,0,0}\makebox(0,0)[b]{\smash{Choreographies}}}%
    \put(0.4930175,0.11714267){\color[rgb]{0,0,0}\makebox(0,0)[b]{\smash{$\epp{\cdot}{}$}}}%
    \put(0.49301747,0.02624387){\color[rgb]{0,0,0}\makebox(0,0)[b]{\smash{$\epp{\cdot}{}$}}}%
    \put(0.31121989,0.21814132){\color[rgb]{0,0,0}\makebox(0,0)[b]{\smash{$\delsel{\cdot}$}}}%
    \put(0.46271787,0.20400151){\color[rgb]{0,0,0}\rotatebox{-90}{\makebox(0,0)[b]{\smash{$\amend$}}}}%
    \put(0.02842361,0.15392827){\color[rgb]{0,0,0}\rotatebox{90}{\makebox(0,0)[b]{\smash{$\delsel{\cdot}$}}}}%
  \end{picture}%
\endgroup%

%% file: minimality.tex
We now discuss our choice of primitives for CC, showing that it is indeed a minimal core language for
choreographic programming.
We first show that if we remove or simplify any primitive from MC we are no longer able to compute all partial
recursive functions using projectable choreographies.
Since label selection can be encoded in MC, we also discuss why it should be included in a core language.
Then we discuss the implications of our results for other choreography languages.

\subsection{Minimality of MC}
\label{sec:minimality}
We proceed by analysing each primitive of MC.
Recall that Turing completeness of MC is a pre-requisite for the Turing completeness of choreography projections.
In most cases, simplifying MC yields a decidable termination problem (thus breaking Turing completeness).
We start with the easiest terms.

\begin{lemma}
  Let $C$ be a choreography with no exit points.
  Then $C$ does not terminate.
\end{lemma}
\begin{proof}
  Straightforward by structural induction on $C$.
\end{proof}

\begin{lemma}
  Let $C$ be a choreography with no communications.
  If $C$ implements a function $f:\NN^n\to\NN$, then, for all inputs $\vec x\in\NN^n$, either $f(\vec x)=x_i$
  for some $i$ or $f(\vec x)$ is undefined.
\end{lemma}
\begin{proof}
  By the semantics of MC, only communication actions can change the state $\sigma$, hence structural induction
  on $C$ shows that $C\sigma\not\to C'\sigma'$ with $\sigma'\neq\sigma$.
  The thesis is a consequence of the definition of function implementation.
\end{proof}

Observe further that the syntax of expressions is trivially minimal: $\zero$ (zero) is the only terminal,
removing $\pcont$ makes termination decidable (since values become statically defined), and likewise for
$\s$ (since no new values can be computed).

\begin{lemma}
  Let $C$ be a choreography with no recursive definitions.
  Then $C$ always terminates.
\end{lemma}
\begin{proof}
  Without recursive definitions, rule \rname{C}{Unfold} is never applicable, hence execution $C$ always
  reduces the size of the choreography.
\end{proof}

Again we observe that recursive definitions are already severely restricted: MC supports only tail recursion
and definitions are not parameterised.

Removing conditionals naturally also breaks Turing completeness.
\begin{theorem}
  \label{lem:nocond-dec}
  Let $C$ be a choreography with no conditionals.
  Then termination of $C$ is decidable and independent of the initial state.
\end{theorem}
\begin{proof}
  The second part is straightforward, since rule~\rname{C}{Cond} is
  the only rule whose conclusion depends on the state.

  For the first part, we reduce termination to a decidable graph
  problem.
  Define $\mathcal G_C=\langle V,E\rangle$ to be the graph whose
  set of vertices $V$ contains $C$ and $\nil$, and is closed under the
  following rules.
  \begin{itemize}
  \item if $\eta;C\in V$, then $C\in V$;
  \item if $\genrec\in V$, then $C_1\in V$;
  \item if $\m{def}\,X=C_2\,\m{in}\,\eta;C_1\in V$,
    then $\genrec\in V$;
  \item if $\m{def}\,X=C_2\,\m{in}\,\eta;X\in V$,
    then $\m{def}\,X=C_2\,\m{in}\,\eta;C_2\in V$.
  \end{itemize}
  This set is finite: all rules add smaller choreographies to $V$,
  except the last one, which can only be applied once for each
  variable in $C$.

  There is an edge between $C_1$ and $C_2$ iff
  $C_1,\sigma\to C_2,\sigma'$ for some $\sigma,\sigma'$ without using
  rule $\rname{C}{Eta-Eta}$.
  This is decidable, as the possibility of a reduction does not depend
  on the state (as observed above).
  Also, if there is a reduction from $C_1$, then there is always an
  edge from $C_1$ in the graph, as swapping communication actions
  cannot unblock execution.

  Then $C$ terminates iff there is a path from $C$ to $\nil$, which
  can be decided in finite time, as $\mathcal G_C$ is finite.
\end{proof}

More interestingly, limiting processes to evaluating only their own local values in conditions makes
termination decidable.
Intuitively, this is because a process can only hold a value at a time and thus no process can compare its
current value to that of another process anymore.

\begin{theorem}
  \label{lem:ifc-dec}
  If the conditional is replaced by $\m{if}\,{\pid p.\pcont=v}\,\m{then}\,{C_1}\,\m{else}\,{C_2}$,
  where $v$ is a value, and rule $\rname{C}{Cond}$ by
\[\infer
{
	\m{if}\,{\pid p.\pcont=v}\,\m{then}\,{C_1}\,\m{else}\,{C_2}, \sigma \ \to \   C_i, \sigma
}
{
	i = 1 \ \text{if } \sigma(\pid p) = v,\
	i = 2 \ \text{otherwise}
}\, ,
\] then termination is decidable.
\end{theorem}

\begin{proof}
  We first show that termination is decidable for processes of the
  form $\m{def}\,X=C_2\,\m{in}\,X$ and comparison with $0$.
  The proof is by induction on the number of recursive definitions in
  $C_2$.

  Consider first the case where $C_2$ has no recursive definitions,
  and let $P$ be the set of all process names occurring in $C_2$.
  We define an equivalence relation on states by
  $$\sigma\equiv_P\sigma'\mbox{ iff }(\forall \pid p\in P,\,\sigma(\pid p)=\varepsilon\mbox{ iff }\sigma'(\pid p)=\varepsilon)\,.$$
  The vertices of the graph are the $2^{|P|}$ equivalence classes of
  states wrt~${\equiv_P}$, plus $\top$.
  Note that ${\equiv_P}$ is compatible with the transition relation
  excluding rule \rname{C}{Eta-Eta}: for any choreography $C$ using
  only process names in $P$, $\sigma_1\equiv\sigma_2$ and
  $C,\sigma_i\to\sigma'_i$, then $\sigma'_1\equiv\sigma'_2$.

  The edges in the graph are defined as follows.
  There is an edge from $[\sigma]$ to $[\sigma']$ if
  $C_2,\sigma\to X,\sigma'$, and there is an edge from $[\sigma]$ to
  $\top$ if $C_2,\sigma\to\nil,\sigma'$ or $C_2,\sigma\to Y,\sigma'$
  for some $Y\neq X$.
  This is constructible, as reductions in $C_2$ are always finite, and
  well-defined, as alternative reduction paths always end in the same
  state. 

  Since reductions are deterministic and ${\equiv_P}$ is compatible
  with reduction, every node has exactly one edge leaving from it,
  except for $\top$.
  Therefore, we can decide if $\m{def}\,X=C_2\,\m{in}\,X$ terminates
  from an initial state $\sigma$ by simply following the path starting
  at $\sigma$ and returning Yes if we reach $\top$ and No if we pass
  some node twice.
  This procedure terminates, as the graph is finite.

  For the inductive step, proceed as above but add an extra node to
  the graph, labeled $\bot$.
  When constructing the edges in the graph, if $C_2$ reduces to a
  variable $Y$ different than $X$, we split into two cases.
  If $Y$ is not bound in $C_2$, we proceed as in the previous case.
  If $Y$ is bound, then we apply the induction hypothesis to the
  choreography $\m{def}\,Y=C_Y\,\m{in}\,Y$ (where $Y=C_Y$ is the same
  as in $C_2$) to decide whether the reduction from $Y$ will
  terminate; if this is not the case, we add an edge to $\bot$,
  otherwise we proceed with the simulation.
  At the end, we return No in the case that the path followed leads
  to~$\bot$.

  The general case follows, as $C$ has the same behaviour as
  $\m{def}\,X=C\,\m{in}\,X$ for some $X$ not occurring in $C$.

  If we allow comparisons with other values, the strategy is the same,
  but the relation $\equiv_P$ has to be made finer.
  The key observation is that only a finite number of values can be
  used in comparisons, so we can identify states if they only differ
  on processes whose contents are larger than all values used in
  conditionals.
\end{proof}

Summarising, simplifying MC in any of the ways described above makes it no longer a representative model
of choreographic programming.

\subsection{CC as a Core Language}
In the Appendix, we formally present an embedding of CC into the choreography model from~\cite{CM13}, which we refer to in this work as Channel Choreographies (ChC).
ChC is a very rich choreography language designed to be projected to a variant of the session-typed $\pi$-calculus~\cite{BCDLDY08}, which we refer to as Channel Processes (ChP).
Communications in ChP are based on channels, instead of process names as in SP.
This layer of indirection means that a process performing an I/O action does not know
\newtext
which other process it is going to communicate with,
\oldtext
and that there can be race conditions on the usage of channels.
ChC comes with a typing discipline for checking that the usage of channels specified in a choreography does
not cause errors in the process code generated by EPP.
In particular, we show that our embedding always yields well-typed and projectable ChC choreographies.

\paragraph{Channel Choreographies}
Our formal translation from CC to ChC shows that many primitives of ChC are not needed to achieve Turing completeness,
including: asynchronous communications, creation of sessions and processes, channel mobility, parameterised
recursive definitions, arbitrary local computation, unbounded memory cells at processes, and multiparty sessions.
While useful in practice, these primitives come at the cost of making the formal treatment of ChC
technically involved.
In particular, ChC (as well as its implementation Chor) requires a sophisticated type system, linearity analysis, and
definition of EPP to ensure correctness of projected processes.
These features are not needed in CC.

\paragraph{Other Choreography Languages}
The language WS-CDL from W3C~\cite{wscdl} and the formal models inspired by it -- e.g.,~\cite{CHY12} -- are very 
similar to ChC, and a similar translation from CC could be formally developed, with similar implications as above. The 
same applies to the choreography language developed in~\cite{PGGLM15}, which adds higher-order features to 
choreographies to achieve runtime adaptation. Finally, the language of compositional choreographies presented 
in~\cite{MY13} is an extension of ChC, and therefore our translation applies directly. This implies that adding 
modularity to choreographies does not add any computational power, as expected.

\paragraph{Process Languages}
Our embedding of CC in ChC identifies a fragment of ChP, via EPP, that is also Turing complete.
This fragment is isomorphic to value-passing CCS~\cite{CCS}: since we only have one channel,
we can interpret the constructs $\cpgensend$ and $\cpgenrecv$ as sending and receiving over a 
channel with name $k\role{AB}$. We thus obtain a deadlock-free and Turing complete fragment of 
value-passing CCS.
\newtext
Deciding whether a given CCS process lies in this fragment is undecidable by Rice's Theorem, but 
it is possible to define a procedure that establishes deadlock-freedom for a large class of such 
processes. Both results are thoroughly discussed in~\cite{ourextractionstuff}, where we also
present a procedure to extract a choreography that represents a particular network in SP.
\oldtext

Since ChC has also been translated to the Jolie programming language~\cite{saveriostuff,MGZ14}, our reasoning
also applies to the latter and, in general, to service-oriented languages based on message correlation. Namely, our 
results identify a deadlock-free and Turing complete fragment of Jolie.

%


%% file: related.tex
\newtext
\paragraph{Choreographies}
The origins of choreographic programming~\cite{M13:phd} stem from the efforts of using 
choreographic descriptions for the specification of interactions among web services. 
In particular, the Web Services Choreography Description Language by the W3C, WS-CDL for short, 
is a choreography language for describing the observable interaction behaviour of web services from 
a global viewpoint~\cite{wscdl}.
The usefulness of formal methods was recognised early on in the development of WS-CDL, and spawned 
a successful line of research based on choreographies. This led to the introduction of the notion 
of EndPoint Projection (EPP), which has been formalised in process 
calculi~\cite{CHY07,QZCY07}, and also adopted in implementations~\cite{savara:website}.

Already in the early works on choreographies, it was evident that choices at the choreography 
level (like our conditionals) play an important role for EPP: their wrong programming may 
lead to projections that behave incorrectly, as we discussed in \S~\ref{sec:selections} 
and \S~\ref{sec:epp}.
The solution of checking whether the other processes could distinguish which branch they should 
execute based on the labels that they receive -- or, equivalently, operation or method names, as in 
service-oriented or object-oriented programming respectively -- was given together with some early 
proposals of EPP~\cite{CHY07,LGMZ08}. The idea of merging was introduced in~\cite{CHY07}, for a 
different calculus than CC.

In the first choreography languages~\cite{wscdl,CHY07}, the construct for 
performing a communication includes both a carried value and a label. This corresponds to having 
both value communication and label selections in a single construct. It is a choice motivated by 
practice, since that is how invocations work in service-oriented computing (where labels are the 
operations offered by a service) and object-oriented programming (where labels are the names of the 
methods offered by an object).

Very soon afterwards, however, a series of choreography models with
separate constructs for value communications and label selections (as we have in 
CC) started emerging. This transition was influenced by the interaction with the 
research line on session types for process calculi, which have the same distinction~\cite{HVK98}. 
There are two main reasons for having this distinction: it makes the model more foundational, since 
each construct is more primitive; and it allows for studying the two primitives separately, which 
is useful since labels are statically defined whereas communicated values are computed at runtime 
(for example, the fact that the two labels are distinct can be statically computed is what allows 
merging to be defined).
The interaction between session types and choreographies spawned a prolific research area where 
choreographies are used as types for protocols in process calculi and concurrent 
languages~\cite{HLVCCDMPRT16,Aetal16}.
The seminal work in this direction is the theory of Multiparty Session Types~\cite{HYC08,HYC16}.
The theory of session types have been recently found to be in a propositions-as-types 
correspondence with linear logic~\cite{CP10,W12}, where the action of making a choice 
(corresponding to a label selection) is again distinct. This correspondence extends 
naturally to choreographies used as types~\cite{CMSY17}.

Thus, selections are important for practical reasons -- for 
example, they model the selection of an operation, or method, from the interface offered by a 
service, or an object -- and for theoretical reasons -- they are primitives identified in 
foundational theories like session types and linear logic. Indeed, all current implementations 
of choreography languages have them~\cite{YHNN13,chor:website,DGGLM17}.

Our encoding from CC to MC and Corollary~\ref{cor:mp-tc} show, for the first time, that as far as 
the computational expressivity of choreography projections is concerned we need only consider value 
communications. Technically, the communication of a choice can be simulated by communicating 
appropriate values and then using local conditionals to understand which choice was transmitted by 
the sender.
This is possible because in choreographic programming we have complete control of the local 
computations performed by each process, differently from works where choreographies are used as 
types (which abstract from computation).

Which model between CC and MC should we then use when studying choreographies?

If the aim is to determine whether a given choreography language is Turing 
complete, then MC is the obvious choice. For example, if a language has (an equivalent version of) 
MC as a fragment, then it is obviously Turing complete. Also, the simple EPP for MC gives a 
procedure for constructing a process implementation for all possible computations.

Another common aim is to observe the communications enacted by a choreography, either statically 
(e.g., for verification) or at runtime (e.g., monitoring). Analysing the flow of choice 
communications is a key element of all behavioural type systems for choreographies to 
date~\cite{CHY12,CM13,MY13,CMS14,G16,chor:website,savara:website}, which aim at checking that a 
choreography correctly implements some protocol specifications.
For these, CC is a representative model that can be used as foundations to start from.
For other kinds of analyses, like information flow in choreographies~\cite{LNN15}, selections are 
not strictly necessary and MC may be the better choice.
Similar considerations apply to monitoring, since typically that involves checking whether the 
actions performed by a choreographed system respect some specifications: if the specifications 
contain information about explicit choice communications, then selections are necessary (CC), 
otherwise MC may offer a simpler base model.
In general, CC is preferrable for all frameworks that have mechanisms where 
operations (or methods) can be selected out of the interface offered by some component (like a 
service or an object).

\paragraph{Choreographic Programming and Applications of our Development}
Essentially, choreographic programming applies the ideas of choreographies and EPP to 
synthesise correct-by-construction implementations of concurrent 
processes.
Languages for choreographic programming are typically more complicated than choreography languages 
for specifications. For example, we have 
choreographic programming languages for: service-oriented computing~\cite{CM13}, including 
notions such as dynamic networks; adaptable computing~\cite{DGGLM17}, including runtime code 
updates; and cyber-physical systems~\cite{LNN16,LH17}, including broadcasts and failures.

Thus, so far, most expressivity results have been investigated for choreographic specifications. For 
example, we know of a strong characterisation result for multiparty session types: a variant of 
multiparty session types corresponds to communicating finite state machines~\cite{BZ83} that 
respect the property of multiparty compatibility~\cite{DY13}.
By contrast, for choreographies used as concrete
implementations (our interest here), this question has barely been scratched before this work:
session-typed choreographies with finite traces correspond to proofs in multiplicative-additive
linear logic~\cite{CMS14}. The language in~\cite{CMS14} does not include any
constructs for programming repetitive behaviour.
To the best of our knowledge, MC is the first choreography language to be identified as minimally 
Turing complete.

Just as it happened for process calculi, the field of choreographic programming is evolving into a 
workshop of different languages developed for different purposes~\cite{M15}. In this context, it 
makes sense to develop new notions by following a minimalistic approach, such that they remain as 
easy to adopt in other choreography models as possible.
The results that we have presented here in extended form have already been useful in this sense. We 
briefly report on some developments.

%
In~\cite{ourasyncstuff}, we showed that CC is unable to encode asynchronous communication, 
and then discussed which primitives to add in order to be able to mimic this behaviour in a 
synchronous semantics.
Thanks to CC, in addition to the technical convenience of working in a setting that is as simple as 
possible, our study again identifies a minimal set of primitives that must be present in a 
choreography language to build such an encoding, and the encoding itself can be structurally 
extended to most choreography languages.

In~\cite{CM17:ice}, we formally defined what a ``good'' asynchronous semantics for 
choreographies is, and show that it is possible to equip MC with such a semantics, obtaining a 
version of MC with an asynchronous reduction semantics.
Again, our construction is simple yet general and is easy to adopt in more sophisticated 
choreography languages.

In~\cite{ourextractionstuff}, we considered the problems of (i) deciding whether a process 
implementation can be described by a choreography, and (ii) synthesising such a choreography in the 
affirmative case; we showed that, for CC and SP, both problems are solvable in exponential time.
Since all languages for choreographic programming introduced so far include CC, the technical 
challenges identified for these problems are present in all such languages, and the 
identified exponential complexity is the current lower bound for all of them.

We have also used CC as a basis to obtain expressive choreography languages for 
more practical purposes. In~\cite{ourPCstuff}, we studied how to extend CC with general sequential 
composition -- a feature that is unavailable in most choreography calculi equipped with recursion -- 
together with other commonly occurring primitives that are useful in practice (process spawning and 
name passing). The resulting calculus is simple, yet expressive enough to capture different 
parallel computing algorithms, like parallel versions of gaussian elimination and fast fourier 
transform~\cite{lcf:fm:16a}.
These two works illustrate how concepts defined for CC (for example, the notion of ``running in 
parallel'', Definition~\ref{defn:parrun}) are naturally applicable in more complex settings.

Our results support the claim that there are substantial advantages to gain from first studying 
choreographic programming in itself, abstracting from features that are specific to other models 
(like channels in process calculi), and then applying the obtained insights to particular scenarios.
We believe that CC is a useful step in this direction, and that it will serve as a stepping stone 
for the future 
developments of the paradigm of choreographic programming in general.

\oldtext


\smallpar{Full $\beta$-reduction and Nondeterminism.}
Execution in CC is nondeterministic due to the swapping of communications allowed by the structural 
precongruence $\precongr$.
This recalls full $\beta$-reduction for
$\lambda$-calculus, where sub-terms can be evaluated whenever possible.
However, the two mechanisms are actually different.
Consider the choreography
$
C \ \defeq \ \com{\pid p.\pcont}{\pid q}; \com{\pid q.\zero}{\pid r};\nil
$.
If CC supported full $\beta$-reduction, we should be able to reduce the second communication before the
first one, since there is no data dependency between the two. Formally, for some $\sigma$:
$
C, \sigma \ \to \ \com{\pid p.\pcont}{\pid q}; \nil, \sigma[\pid r \mapsto \zero]
$.
However, this reduction is 
disallowed by our semantics: rule $\rname{C}{Eta-Eta}$ cannot be applied because $\pid q$ is present in both
communications.
This difference is a key feature of choreographies, stemming from their practical origins: controlling
sequentiality by establishing causalities using process identifiers is important for the implementation
of business processes~\cite{wscdl}.
For example, imagine that the choreography $C$ models a payment transaction and
that the message from $\pid q$ to $\pid r$ is a confirmation that $\pid p$ has sent its credit card information to
$\pid q$; then, it is a natural requirement that the second communication happens only after the first.
Note that we would reach the same conclusions even if we adopted an asynchronous messaging semantics for SP, since the
first action by $\pid q$ is a blocking input.

While execution in CC can be nondeterministic, computation results are deterministic as in many other choreography 
languages~\cite{CM13,CMS14,MY13}:
if a choreography terminates, the result will always be the same regardless of how its execution is scheduled,
recalling the Church--Rosser Theorem for the $\lambda$-calculus~\cite{CR36}.
Nondeterministic computation is not necessary for our results. Nevertheless, it can be easily added to CC.
Specifically,
we could augment CC with the syntax primitive
$C_1 \oplus^{\pid p} C_2$
and the reduction rule $C_1 \oplus^{\pid p} C_2 \ \to \ C_i$ for $i=1,2$.
Extending SP with an internal choice $B_1 \oplus B_2$ and our definition of EPP is straightforward:
in SP, we would also allow $B_1 \oplus B_2 \ \to \ B_i$ for $i=1,2$, and define
$\epp{C_1 \oplus^{\pid p} C_2}{\pid r}$
to be
${\epp{C_1}{\pid r}}\oplus {\epp{C_2}{\pid r}}$
if $\pid r=\pid p$ and
$\epp{C_1}{\pid r} \merge \epp{C_2}{\pid r}$
otherwise.

\smallpar{Merging and Amendment.}
Amendment was first studied by~\citet{LMZ13} for
a simple language with finite traces (thus not Turing complete). Our definition is different, since it uses merging for 
the first time.

%
%
We could define our amendment procedure in different ways, e.g., by propagating
 selections from a process to another as a chain, rather than from one
 process to all the others. This would not influence our results.
%

%
\smallpar{Actors and Asynchrony.}
Processes in SP communicate by using direct references to each other, recalling actor systems.
However, there are notable differences: communications are synchronous and
inputs specify the intended sender. The first difference comes from minimality:
asynchrony would add possible behaviours to CC, which are unnecessary to establish Turing completeness. We 
leave an investigation of asynchrony in CC to future work.
The second difference arises because CC is a choreography calculus, and communication primitives in choreographies 
typically express both sender and receiver.

\paragraph{Register Machines}
The computational primitives in CC recall those of the
Unlimited Register Machine (URM)~\cite{Cutland80},
but CC and URM differ in two main aspects.
First, URM programs contain go-to statements, while CC supports only tail recursion.
Second, in the URM there is a single sequential program manipulating the cells, whereas
in CC computation is distributed among the various cells (the
processes), which operate concurrently.

Simulating the URM is an alternative way to prove Turing completeness of CC.
However, our proof using partial recursive functions is more direct and gives an algorithm to 
implement any function
in CC, given its proof of membership in $\RR$.
It also yields the natural interpretation of parallelisation stated
in Theorem~\ref{thm:par-compute}.
Similarly, we could establish Turing completeness of CC using only a bounded number of
processes. However, such constructions encode data using G\"odel numbers, which is not in the spirit 
of our declarative
notion of function implementation. They also restrict concurrency, breaking 
Theorem~\ref{thm:par-compute}.


%% file: appendix.tex
\newtext
\section{CC as a Core Language: Channel Choreographies}
\oldtext
\label{sec:translation}

CC is representative of the body of previous work on choreographic programming, where choreographies are
used for implementations, for example~\cite{CHY12,CM13,chor:website,MY13,PGGLM15,wscdl}.
All the primitives of CC (and therefore of MC) can be encoded in such languages.
Thus, we obtain a notion of function implementation for these languages, induced by that for CC, for which they
are Turing complete.

In this section we make this claim precise for the model in~\cite{CM13}, which we refer to in this work as
Channel Choreographies (ChC).
ChC is designed to be projected to a variant of the session-typed $\pi$-calculus~\cite{BCDLDY08}, which we
refer to as Channel Processes (ChP).
Communications in ChP are based on channels, instead of process names as in SP.
This layer of indirection means that a process performing an I/O action does not know
\newtext
which other process it is going to communicate with,
\oldtext
and that there can be race conditions on the usage of channels.
ChC comes with a typing discipline for checking that the usage of channels specified in a choreography does
not cause errors in the process code generated by EPP.

\subsection{Channel Choreographies}


\paragraph{Syntax}
We report the full syntax of ChC in Figure~\ref{fig:chc_syntax}.
%
\input{chc_syntax}
Several terms are unnecessary for our translation; we \lgbox{box} such terms in our presentation of the
syntax.
In the original presentation of ChC, expressions $e$ may contain any basic values (integers, strings, etc.) or
computable functions, making the language trivially Turing complete. Also, labels $l$ range over an infinite set. Here, 
for our development, we need only to consider expressions of the form $\zero$ or $\suc x$, and labels $\lleft$ and 
$\lright$ (as in CC).
The major difference between CC and ChC is the usage of \emph{public channels} $a$ and \emph{session channels}
$k$.
Public channels are used to create new processes and channels at runtime, whereas session channels are used
for point-to-point communications between processes.
We only need a single session channel in our development.

An interaction $\eta$ in ChC can be either a start, a value communication, a selection, or a delegation.
In a start term $\scgenstart$, the processes $\pids p$ on the left synchronise at the public channel $a$ in
order to create a new private session $k$ and spawn some new processes $\pids q$ ($k$ and $\pids q$ are bound
to the continuation).
Each process is annotated with the role it plays in the created session.
Roles are ranged over by $\role A, \role B, \role C, \ldots$.
They are used in the typing discipline of
\newtext
ChC
\oldtext
to check whether sessions are used according to protocol
specifications, given as multiparty session types~\cite{HYC16}.

In a value communication $\scgencom$, process $\pid p$ sends its evaluation of expression $e$ over session $k$ to 
process $\pid q$, which stores the result in its local variable $x$; the name $x$ appearing under $\pid q$ is bound to 
the continuation.
Differently from CC, where each process has only one memory cell accessed through the placeholder $\pcont$, in
ChC each process has an unbounded number of cells (variables).
Selections in ChC, of the form $\scgensel$, are very similar to those in CC: the only difference is that we
also have to write which role each process plays and the session used for communicating.
In a delegation term $\scgendel$, process $\pid p$ delegates its role $\role C$ in session $k'$ to process
$\pid q$; delegation in ChC is a typed form of channel mobility, inspired by the $\pi$-calculus.

In a conditional, process $\pid p$ chooses a continuation based on whether the expressions $e$ and $e'$
evaluate to the same value according to its own local state.
The restriction term $\res r C$ is standard and binds the scope of $r$ (which can be either a process name
$\pid p$ or a session channel name $k$) to $C$.
Finally, in the definition of a recursive procedure, the parameters $\til D$ indicate which processes are used
in the body of the procedure and which variables and sessions are used by each process.
In the invocation of a procedure $\scgencall$, each process can pass generic expressions as parameters to itself.

\paragraph{Semantics}
ChC was originally presented with an asynchronous semantics~\cite{CM13}.
We first present our results using only the (simpler) synchronous variant of the semantics of ChC, and defer the discussion
of the general asynchronous case to the end of this section.
This semantics is given in terms of a reduction relation, presented in Figure~\ref{fig:chc_semantics}.
%
\input{chc_semantics}
Rule $\rname{Ch}{Com}$ is the key rule, where the value sent from a process $\pid p$ is received by a process
$\pid q$.
Technically, this is modelled by replacing variable $x$ with $v$ in the continuation $C$, but only when it
appears under the process name $\pid q$ (the \emph{smart substitution} $C[v/x@\pid q]$).
Rule $\rname{Ch}{Cond}$ models an internal choice: $\pid p$ chooses a continuation depending on whether the two
values $v$ and $w$ are the same.
Rules $\rname{Ch}{Del}$ and $\rname{Ch}{Start}$ implement the informal semantics of delegation and start described 
earlier; we do not use them in our development.
The other rules are similar to those of CC.
The structural precongruence $\precongr$ is defined as expected, following the same intuition as that for CC.
In particular, it supports swapping two terms whenever they involve disjoint process names.

As expected, ChC offers a deadlock-freedom-by-design property in the style of
Theorem~\ref{thm:df-by-design}~\cite{CM13}.

\subsection{Channel Processes (ChP)}
We now present Channel Processes (ChP), the target language that choreographies in ChC can be projected to.
We discuss only the terms used in our work (see~\cite{CM13} for a complete presentation).

\paragraph{Syntax}
The relevant part of the syntax of processes ($P,Q$) is reported in Figure~\ref{fig:chp_syntax}.
Binding occurrences are denoted by the usage of round parentheses.
%
\input{chp_syntax}
Terms $\cpgenreq$, $\cpgenacc$ and $\cpgenserv$ are used to start a new session $k$ by synchronising on the
public channel $a$, and model respectively: the process requesting the creation of the session (responsible
for playing the first role in $\roles A$); a process accepting to play role $\role A$ in the session; and,
finally, a replicated process that will spawn a fresh process for playing role $\role A$.
In the first line we have the terms for in-session communications.
In term $\cpgensend;P$, as role $\role A$ on session $k$, we send the value of expression $e$ to $\role B$ on
the same session; then, we proceed as $P$.
Dually, term $\cpgenrecv;P$ receives a message for role $\role B$ from role $\role A$ on session $k$ and
stores it in variable $x$.
Terms $\cpgensel$ and $\cpgenbranch$ model, respectively, branch selection and offering.
Finally, terms $\cpgendels$ and $\cpgendelr$ capture channel mobility.
The other terms are the standard parallel composition, procedure definition, procedure call, conditional
(restricted to checking for equality), and terminated process.

\paragraph{Semantics}
As before, we discuss only the synchronous semantics of ChP.
We discuss only the communication rules, shown in Figure~\ref{fig:chp_semantics}, as all the other rules are
standard -- see~\cite{KY13}.
As in typical calculi for multiparty sessions equipped with roles, each role in a session is a distinct
communication endpoint.
Therefore, a send action on a session $k$ from a role $\role A$ towards a role $\role B$ synchronises with a
receive action on the same session $k$ by the target role $\role B$ wishing to receive from the sender role
$\role A$.
%
\input{chp_semantics}

\subsection{Endpoint Projection and Typing}

As for CC, the Endpoint Projection from ChC to ChP is defined by first defining how to project the behaviour
of a single process.
The projection of a process $\pid p$ from a choreography $C$, written $\epp{C}{\pid p}$, is inductively
defined on the structure of $C$ in a similar way as the behaviour projection given in \S~\ref{sec:epp}.
The complete EPP procedure from ChC to ChP is technically involved, because the start term $\scgenstart$ found
in ChC enables the reuse of the same services exposed at a public channel $a$ for spawning processes with
potentially different behaviour.
However, since start terms and restriction of names are unnecessary for our development, we can use a much
simpler definition -- see~\cite{CM13} for the general case.
We report the rules for projecting value communications and conditionals in Figure~\ref{fig:chc_epp}.
%
\input{chc_epp}
The merging operator $P \merge Q$ works as in CC: it is isormorphic to $P$ and $Q$ aside from input
branches with distinct labels, which are instead included in a larger input branching.


\subsection{Typing ChC}
Differently from CC, the EPP of a choreography in ChC does not always yield correct results.
Consider the following choreography:
\[
C \ = \ \scsel{\pid p[\role A]}{\pid q[\role B]}{k}{\lleft};\  \sccom{\pid q[\role B].\zero}{\pid p[\role A].x}{k};\ \scsel{\pid r[\role A]}{\pid q[\role B]}{k}{\lleft}
\]
The choreography $C$ above always terminates by reaching $\nil$ (by using rules $\rname{Ch}{Sel}$, then
$\rname{Ch}{Com}$, and then $\rname{Ch}{Sel}$ again).
However, its EPP (albeit defined) may get stuck:
\begin{displaymath}
\begin{array}{l@{\ }c@{\quad }c@{}c}
& \epp{C}{\pid p} & & \epp{C}{\pid r}\\
\epp{C}{} =  & \overbrace{\cpsel{k}{\role A}{\role B}{\lleft};\ \cprecv{k}{\role A}{\role B}{x}}
& \parp & \overbrace{ \cpsel{k}{\role A}{\role B}{\lleft} } 
\\[2mm]
& \multicolumn{3}{c}{
\!\!\! \!\!\parp \ \underbrace{ \cpbranch{k}{\role B}{\role A}{\{ \lleft:\ \cpsend{k}{\role B}{\role A}{\zero}; \
\cpbranch{k}{\role B}{\role A}{\{ \lleft: \nil \} \}} } }
}
\\
& \multicolumn{3}{c}{\phantom{\parp \ } \epp{C}{\pid q} }
\end{array}
\end{displaymath}

Above, we have a race between the projections of process $\pid p$ and process $\pid r$ for the selection of
label $\lleft$ offered by process $\pid q$.
This is because both $\pid p$ and $\pid r$ play the same role $\role A$ in session $k$ and therefore the
receiver (the projection of process $\pid q$) cannot distinguish them.
In the case where the race is won by the projection of process $\pid r$, not only do we obtain a reduction not
defined by the originating choreography, but we even get into a deadlocked situation:
\[
\epp{C}{}\ \to\  \cpsel{k}{\role A}{\role B}{\lleft};\ \cprecv{k}{\role A}{\role B}{x} \ \parp \ \cpsend{k}{\role B}{\role A}{\zero}; \
\cpbranch{k}{\role B}{\role A}{\{ \lleft: \nil \} } 
\]
To avoid such situations, ChC comes with a typing discipline based on multiparty session types that guarantees
the absence of races.

A typing judgement for CC has the form $\tjudge{\Gamma;\Theta}{C}{\Delta}$, where $\Delta$ types the usage of
sessions, $\Theta$ the ownership of roles by processes, and $\Gamma$ variables and public channels.

Formally, the typing environment $\Gamma$ contains variable typings of the form $x@\pid p:S$, typing variable
$x$ at $\pid p$ with data type $S$ (which can only be $\tnat$ in our case).
An environment $\Theta$ contains ownership typings of the form $\pid p:k[\role A]$, read ``process $\pid p$
owns role $\role A$ in $k$'' (when writing $\Theta,\pid p:k[\role A]$, it is assumed that no other process
owns the same role for the same session in $\Theta$).
The environment $\Delta$ contains session typings of the form $k:G$, where $G$ is a global type~\cite{HYC16}.
The syntax of global types is given in Figure~\ref{fig:gt_syntax}.
%
\input{gt_syntax}
A global type $G$ abstracts a communication between two roles in a session.
A value communication is abstracted by $\gtgencom$ (we restrict values to be natural numbers).
A global type $\gtgenbranch$ allows any selection from $\role A$ to $\role B$ of one of the labels $l_i$,
provided that then the session proceeds as specified by the corresponding continuation $G_i$.
The other terms are for recursion ($\gtgenrec$ and $\gtgencall$) and termination ($\gtend$).

We discuss the most relevant typing rules for ChC, given in Figure~\ref{fig:chc_types}.
\input{chc_types}
Rule $\rname{T}{Com}$ checks that, in a value communication on session $k$, the sender and receiver processes
own their respective roles in session $k$ ($\Theta \seq \pid p: k[\role A], \pid q: k[\role B]$), that the
protocol for session $k$ expects a communication for their respective roles ($k: \gtcom{\role A}{\role B}{S};
G$), and that the expression sent by the sender has the expected type $S$.
Rule $\rname{T}{Sel}$ checks that a selection uses one of the labels expected by the protocol for the session
($j\in I$).
Rule $\rname{T}{Cond}$ is standard, requiring both branches to have the same typing; observe that different
communication behaviour in the two branches may still occur, because of rule $\rname{T}{Sel}$.
Rules $\rname{T}{Call}$ and $\rname{T}{Def}$ type, respectively, recursive calls and recursive procedures. These rules 
are simplified compared to the presentation in~\cite{CM13}, taking into account that our encoding always calls 
procedures with exactly the same arguments (processes and variables) as they are declared.

Well-typedness is preserved by reductions.
Furthermore, using this type system we get an operational correspondence result for EPP from ChC to ChP.
\begin{theorem}[Operational Correspondence (ChC $\leftrightarrow$ ChP)~\cite{CM13}]
  \label{thm:oc-chc-chp}
  Let $C$ be a well-typed channel choreography without start subterms (terms of the form $\scgenstart$) and
  such that its endpoint projection $\epp{C}{}$ is defined.
  Then:
  \begin{itemize}
  \item (Completeness) $C \to C'$ implies $\epp{C}{} \to \succ \epp{C'}{}$;
  \item (Soundness) $\epp{C}{} \to P$ implies $C \to C'$ and $\epp{C'}{} \prec P$.
  \end{itemize}
  where $\prec$ is the pruning relation defined in~\cite{CM13}.
\end{theorem}


As for CC, 
the EPP of a well-typed channel choreography never deadlocks.

\subsection{Embedding CC into ChC}

Defining an embedding from CC to ChC is nontrivial, as the communication primitives of CC and ChC are
different.
In CC, messages are passed directly between processes: each process knows whom it is sending to or receiving
from in each communication step; in ChC, communication is between roles in a session channel.
To translate core choreographies into channel choreographies, we therefore assign to each process a role
syntactically identical to its name, and perform all communication over a fixed channel $k$.

Conditional terms are also not directly translatable, as ChC evaluates guards in a single process.
For this reason, each translated process uses two variables: $x$, storing its internal value, and $y$, used
exclusively for temporary storage of a value required for a test.

For recursion, we recall that  $\pn(C)$ returns the set of process names in $C$.

\begin{definition}[Embedding of CC in ChC]
  The \emph{embedding} of a core choreography $C$ in ChC is $\emb{C}$, inductively defined as follows.
  \begin{align*}
    \emb{\gencom;C} &= \com{\pid{p[p]}.e[x/\pcont]}{\pid{q[q]}.x}:k;\emb{C} \\
    \emb{\gensel;C} &= \com{\pid{p[p]}}{\pid{q[q]}}:k[l];\emb{C} \\
    \emb{\cond{\eqcom{\pid p}{\pid q}}{C_1}{C_2}} &= \com{\pid{q[q]}.x}{\pid{p[p]}.y}:k;\\
      & \qquad\m{if}\,{\pid p.(x = y)}\,\m{then}\,{\emb{C_1}} \,\m{else}\,{\emb{C_2}} \\
    \emb{\rec{X}{C_2}{C_1}} &= \big(\m{def}\,{X(\ast)} = {\emb{C_2}}\big)[{\pnast{C_2}}/\ast]\, \m{in}\,{\emb{C_1}} \\
    \emb\gencall=\call X\langle\ast\rangle
    \qquad&\qquad
    \emb\nil=\nil
  \end{align*}
  where $\pnast A=\{\pid p(\{x,y\},k)\mid\pid p\in\pn(A)\}$.
\end{definition}


\begin{lemma}
  \label{lem:cong}
  Let $C$ and $C'$ be core choreographies.
  Then $C \precongr C'$ if and only if $\emb{C} \precongr \emb{C'}$.
\end{lemma}
\begin{proof}
  For the direct implication, observe that all structural precongruence rules in CC become valid instances of
  precongruence in ChC when mapped by $\emb\cdot$.
  Conversely, given a structural precongruence rule in ChC, if its arguments are in the image of $\emb\cdot$,
  then the rule can be pulled back to a valid precongruence in CC.
\end{proof}

In order to compare the semantics of core and channel choreographies, we need to take the state into account.
This is done by viewing each state as a substitution, replacing all free occurrences of $x$ with the actual
content of the process it belongs to.

\begin{definition}[Substitution induced by state]
  Let $C$ be a core choreography and $\sigma$ be a state.
  The substitution $\sigma_C$ is defined as $\sigma_C=[\sigma(\pid p)/x@\pid p\mid\pid p\in\pn(C)]$, and the
  embedding of $C$ in ChC via $\sigma$ is the channel choreography $\emb C_\sigma=\sigma_C(\emb C)$.
\end{definition}

Below, $\to^+$ denotes a chain of one or more applications of $\to$, and $\to^?$ denotes identity or one
application of $\to$.
\begin{theorem}[Operational Correspondence (CC $\leftrightarrow$ ChC)]
  \label{thm:oc-cc-chc}
  Let $C$ be a choreography in CC. Then, for all $\sigma$:
  \begin{itemize}
  \item\emph{(Completeness)} $C,\sigma \to C',\sigma'$ implies $\emb{C}_\sigma\to^+ \emb{C'}_{\sigma'}$;
  \item\emph{(Soundness)} $\emb{C}_\sigma\to C'$ implies $C,\sigma \to C^\ast,\sigma^\ast$ and
    $C'\to^?\emb{C^\ast}_{\sigma^\ast}$.
  \end{itemize}
\end{theorem}
\begin{proof}
  \begin{itemize}
  \item\emph{(Completeness)} We analyze the possible cases for the rule justifying $C\lto\eta C'$.
    %
    The only non-trivial case is rule \rname{C}{Com}.
    Suppose $C$ is $\gencom;C'$.
    If $C,\sigma\to C',\sigma'$, then
    $\emb{C}_\sigma=\com{\pid p.\sigma_C(e[x/\pcont])}{\pid q.x:k};\emb{C'}_{\sigma''}$ where
    $\sigma''=\sigma_C\setminus\{\sigma(\pid q)/x@\pid q\}$ can make a transition to
    $\emb{C'}_{\sigma''}[e[\sigma_C(\pid p)/\pcont]/x@\pid q]$, which coincides with $\emb{C'}_{\sigma'}$.

  \item\emph{(Soundness)} The proof is again by case analysis on the transition from $\emb{C}_\sigma$ to $C'$,
    noting that this cannot involve delegation or start actions.
    Most cases are straightforward, except when the transition is a communication obtained from translating a
    conditional.
    In this case, $C'$ must execute the full conditional action, and $\emb{C^\ast}_{\sigma^\ast}$ reduces to
    $C'$ by application of rule \rname{Ch}{Cond}.
  \end{itemize}
\end{proof}

We now define a notion of function implementation in ChC.
Since the semantics of ChC does not have state, this definition is slightly different than that for CC.
The embedding of CC is then a Turing complete fragment of ChC, which we show to be projectable.

\begin{definition}[Implementation in ChC]
  A channel choreography $C$ \emph{implements} a function $f:\NN^n\to\NN$ with input variables $\pid
  p_1.z_1$,\ldots,$\pid p_n.z_n$ and output variable $\pid q.z$ if, for all $x_1,\ldots,x_n\in\NN$:
  \begin{itemize}
  \item if $f(\til x)$ is defined, then
    $C[\wtil{\numeral{x_i}/z_i@\pid p_i}]\to^\ast\nil$, and $\pid q$ receives
    exactly one message with $\numeral{f(\til x)}$ as the value
    transmitted;
  \item if $f(\til x)$ is not defined, then
    $C[\wtil{\numeral{x_i}/z_i@\pid p_i}]\not\to^\ast\nil$, and $\pid q$ never
    receives any messages.
  \end{itemize}
\end{definition}

\begin{theorem}[Soundness]
  \label{thm:sound2}
  If $f:\NN^n\to\NN$ is a partial recursive function, then $\emb{\enc f{\pids p\mapsto\pid q}{}}$ implements
  $f$ with input variables $\wtil{\pid p.x}$ and output variable $\pid q.x$.
\end{theorem}
\begin{proof}
  Consequence of the proof of Theorem~\ref{teo:sound1} and of Theorem~\ref{thm:oc-cc-chc}: as the only free
  variables in $\emb{\enc f{\pids p\mapsto\pid q}{}}$ are $\pid p_i.x$ for $1\leq i\leq n$,
  $\emb{\enc f{\pids p\mapsto\pid q}{}}[\numeral{x_i}/z_i@\pid p_i]$ coincides with
  $\emb{\enc f{\pids p\mapsto\pid q}{}}_\sigma$ whenever $\sigma$ contains $\numeral{x_i}$ at each process
  $\pid p_i$.
\end{proof}

\newtext
We conclude by combining our results to characterise a Turing-complete and deadlock-free fragment of ChP.
\oldtext

Let \ChPP{} be the smallest fragment of ChP containing the projections of all typable and projectable
choreographies in ChC, formally: $\ChPP{}=\{\epp{C}{}\mid\epp{C}{}\mbox{ is defined}\}$.
By Theorem~\ref{thm:oc-chc-chp}, all terms in \ChPP{} are deadlock-free.

We now show that \ChPP{} is also Turing powerful.
The development is similar to that for \CPP{} (\S~\ref{sec:sound}), but we need two additional steps.
First, the operational correspondence theorem for the EPP of ChC (Theorem~\ref{thm:oc-chc-chp}) needs the
projected channel choreography to be well-typed.
Fortunately, this is always the case for the channel choreographies obtained by embedding amended CC terms.

\begin{lemma}
  Let $C$ be a core choreography and $\sigma$ a state.
  Then $C' = \emb{\amend(C)}_\sigma$ implies $\tjudge{\Gamma;\Theta}{C'}{\Delta}$ for some $\Gamma$, $\Theta$
  and $\Delta$.
\end{lemma}
\begin{proof}
  Choosing $\Theta$ is trivial, as each process has its own role.
  For $\Gamma$, we assign type $\tnat$ to all variables.
  Finally, $\Delta = k:G$, where $G$ is inferred by abstracting the communications in $C$.
  The inductive construction of the latter is always possible since we applied $\amend$, so we can type each
  conditional with either the same global type or a branching global type with two labels.
\end{proof}

Second, we need to know that the embedding of a projectable core choreography is also projectable in ChC.
\begin{lemma}
  \label{lem:emb_a_projectable}
  If $C$ is projectable, then $\emb{C}_{\sigma}$ is projectable for any $\sigma$.
\end{lemma}
\begin{proof}
  The thesis follows from the fact that: if the projections of two choreographies are mergeable in CC, then
  the projections of their embeddings into ChC are mergeable.
  This is proven by structural induction.
\end{proof}

Using these results, the proof of Corollary~\ref{cor:cp-tc} can be adapted to yield the following property.
\begin{corollary} [Turing completeness of \ChPP{}]
  Every partial recursive function is implementable in \ChPP{}.
\end{corollary}
We thus characterise a fragment of the session-based $\pi$-calculus from~\cite{BCDLDY08} that contains only
deadlock-free terms and is Turing complete.

\subsection{Asynchronous semantics}

The original semantics of ChC is made asynchronous by means of an additional rule that allows some transitions
protected by a prefix to be executed, yielding more possible reduction sequences.
In this scenario, the operational correspondence between CC and ChC is no longer as strong as stated in
Theorem~\ref{thm:oc-cc-chc}; in particular, the result for soundness now reads
\begin{quote}
  \emph{(Soundness)} $\emb{C}_\sigma\to^+ C'$ implies $C,\sigma \to^+ C^\ast,\sigma^\ast$ and
  $C'\to^*\emb{C^\ast}_{\sigma^\ast}$
\end{quote}
However, the reduction relation in the asynchronous setting, restricted to the language fragment we consider
(embeddings of core choreographies, which in particular are typable) is still confluent, and it includes all
synchronous executions.
Therefore, Theorem~\ref{thm:sound2} remains valid in this general case.

%% file: chc_syntax.tex
\begin{figure}
\begin{align*}
C & ::= \eta; C \quad \mid \quad \scgencond \quad \mid \quad \nil
\\
& \phantom{::=} \mid\  \scgenrec \quad \mid \quad \scgencall \quad \mid\quad \mlgbox{\res r C}
\\[1ex]
\eta & ::= \mlgbox{\scgenstart} \quad \mid \quad \scgencom 
\\
& \phantom{::=} \mid\  \scgensel \quad \mid \quad \mlgbox{\scgendel}
\\[1ex]
D & ::= \pid p(\til x,\til k)
\qquad\qquad
E ::= \pid p(\til e, \til k)
\end{align*}
\caption{Channel Choreographies, Syntax.}
\label{fig:chc_syntax}
\end{figure}

%% file: chc_semantics.tex
\begin{figure}
\begin{eqnarray*}
&\infer[\rname{Ch}{Com}]
{
	\scgencomv;C \
	\to
	\
	C[v/x@\pid q]
}
{
}
\qquad
\infer[\rname{Ch}{Sel}]
{
	\scgensel;C \ \to \ C
}
{}
\\[1ex]
&\mlgbox{
  \infer[\rname{Ch}{Del}]
{
	\scgendel; C \ \to \ C
}{}
}
\qquad
\mlgbox{
\infer[\rname{Ch}{Start}]
{
	\scgenstart; C \ \to \ \res{\pids q,k} C
}
{}
}
\\[1ex]
&\infer[\rname{Ch}{Cond}]
{
	\scgencondv \ \to \   C_i
}
{
	i = 1 \ \text{if } v = w,\ 
	i = 2 \ \text{otherwise}
}
\\[1ex]
&\infer[\rname{Ch}{Struct}]
{
	C_1\ \to \  C'_1
}
{
	C_1\, \precongr \,C_2
	& C_2 \ \to \  C'_2
	& C'_2 \, \precongr\, C'_1
}
\\[1ex]
&\infer[\rname{Ch}{Ctx}]
{
	\scgenrec \ \to \ 
	\rec{X(\til D)}{C_2}{C'_1}
}
{
	C_1 \ \to \   C'_1
}
\end{eqnarray*}
\caption{Channel Choreographies, Semantics.}
\label{fig:chc_semantics}
\end{figure}

%% file: chp_syntax.tex
\begin{figure*}
\begin{align*}
P,Q & ::= \quad
\cpgensend;P \cpspar  \cpgenrecv;P
\cpspar
\cpgensel;P \cpspar  
\cpgenbranch
\\
& \phantom{::=} |\quad\  P \parp Q \ \cpspar\  \cpgencond
\cpspar \cpgenrec \cpspar \cpgencall \cpspar 
\nil
\end{align*}
\caption{Channel Processes, Syntax (selection).}
\label{fig:chp_syntax}
\end{figure*}

%% file: chp_semantics.tex
\begin{figure*}
\begin{eqnarray*}
  &\infer[\rname{CP}{Com}]
        {
	\cpgensendv;P \ \parp \ \cpgenrecv;Q
	\quad \to \quad
	P \ \parp \ Q[v/x]
        }{}
\\[1ex]
&\infer[\rname{CP}{Sel}]
      {
	\cpgenselj;P \ \parp \ \cpgenbranchQ
	\quad\!\! \to \quad\!\!
	P \ \parp \ Q_j
      }{
	(j \in I)
      }
\end{eqnarray*}
\caption{Channel Processes, Semantics (selection).}
\label{fig:chp_semantics}
\end{figure*}

%% file: chc_epp.tex
\begin{figure*}
\begin{align*}
&\epp{\scgencom;C}{\pid r} =
	\begin{cases}
		\cpgensend;\epp{C}{\pid r} & \text{if } \pid r = \pid p \\
		\cpgenrecv;\epp{C}{\pid r} & \text{if } \pid r = \pid q \\
		\epp{C}{\pid r} & \text{otherwise}
	\end{cases}
\\[1ex]
&\left[\!\!\left[\begin{array}{l}
      \m{if}\ \pid p.(e = e')
      \\ \m{then}\ C_1\ \m{else}\ C_2
    \end{array}\right]\!\!\right]_{\pid r} =
  \begin{cases}\!\!\!
    \begin{array}{l}
      \cond{e=e'\\\!}{\epp{C_1}{\pid r}}{\epp{C_2}{\pid r}}
    \end{array} & \text{if } \pid r = \pid p \\\\[-2mm]
    \epp{C_1}{\pid r} \merge \epp{C_2}{\pid r} & \text{otherwise}
  \end{cases}
\end{align*}
\caption{Channel Choreographies, EndPoint Projection (relevant cases).}
\label{fig:chc_epp}
\end{figure*}

%% file: gt_syntax.tex
%
\begin{figure}
\begin{align*}
G & ::= \gtgencom; G \ | \ \gtgenbranch \ | \  \gtgenrec;G \ | \ \gtgencall \ | \  \gtend
\\[2mm]
S & ::= \tnat \  | \ \tstring \ | \ \ldots \qquad\qquad l ::= \lleft \ |\  \lright
\end{align*}
\caption{Global Types, Syntax.}
\label{fig:gt_syntax}
\end{figure}

%% file: chc_types.tex
\begin{figure}
\begin{eqnarray*}
&\infer[\rname{T}{Com}]
{
	\tjudge{\Gamma;\Theta}{\scgencom;C}{\Delta, k: \gtcom{\role A}{\role B}{S}; G}
}
{
	\Gamma \seq e@\pid p: S
	\!&\!
	\Theta \seq \pid p: k[\role A], \pid q: k[\role B]
	\!&\!
	\tjudge{\Gamma,x@\pid q:S;\Theta}{C}{\Delta, k: G}
}
\\[1ex]
&\infer[\rname{T}{Sel}]
{
	\tjudge{\Gamma;\Theta}{\scgenselj;C}{\Delta, k: \gtgenbranch}
}
{
	\Theta \seq \pid p: k[\role A], \pid q: k[\role B]
	&
	j \in I
	&
	\tjudge{\Gamma;\Theta}{C}{\Delta, k: G_j}
}
\\[1ex]
&\infer[\rname{T}{Cond}]
{
	\tjudge{\Gamma;\Theta}{\scgencond}{\Delta}
}
{
	\tjudge{\Gamma;\Theta}{C_1}{\Delta}
	&
	\tjudge{\Gamma;\Theta}{C_2}{\Delta}
}
\\[1ex]
&\infer[\rname{T}{Call}]{
	\tjudge{\Gamma,X(\til D):(\Gamma';\Theta;\Delta');\Theta}{X\langle\til D\rangle}{\Delta,\Delta'}
}{\Delta \mbox{ \m{end} only}
&
\Gamma' \subseteq \Gamma
}
\\[1ex]
&\infer[\rname{T}{Def}]{
	\tjudge{\Gamma;\Theta}{\scgenrec}{\Delta}
}{
\begin{array}{ll}
	\tjudge{\Gamma,X(\til D):(\Gamma';\Theta';\Delta'); \Theta}{C_1}{\Delta} &
	\Gamma' \subseteq \Gamma
	\\[1ex]
	\tjudge{\Gamma', X(\til D):(\Gamma';\Theta';\Delta'); \Theta'}{C_2}{\Delta'} &
	\Theta' \subseteq \Theta
	\end{array}
}
\end{eqnarray*}
\caption{Channel Choreographies, Typing Rules (selection).}
\label{fig:chc_types}
\end{figure}